\documentclass[12pt, draftclsnofoot, onecolumn, letter]{IEEEtran}
\newlength{\figurewidth}
\usepackage{graphicx}
\graphicspath{{./figures/}}
\usepackage[usenames,dvipsnames]{xcolor}

\usepackage{amsmath,amssymb,amsthm,amsfonts}
\usepackage{nicefrac}

\usepackage{cite}
\usepackage[footnote,draft,silent,nomargin]{fixme}
\usepackage[iso,german]{isodate}
\usepackage{color}
\usepackage{transparent}

\usepackage{soul}

\ifCLASSOPTIONcompsoc
\usepackage[caption=false,font=normalsize,labelfont=sf,textfont=sf]{subfig}
\else
\usepackage[caption=false,font=footnotesize]{subfig}
\fi

\newcounter{tempEquationCounter}
\newcounter{thisEquationNumber}

\allowdisplaybreaks

\IEEEoverridecommandlockouts 
\begin{document}

\title{Delay Performance of Wireless Communications with Imperfect CSI and Finite Length Coding}

\author{Sebastian~Schiessl,~\IEEEmembership{Student Member,~IEEE,}
Hussein~Al-Zubaidy,~\IEEEmembership{Senior~Member,~IEEE,}
Mikael~Skoglund,~\IEEEmembership{Senior~Member,~IEEE,}
and~James~Gross,~\IEEEmembership{Senior~Member,~IEEE}
\thanks{The authors are with the School of Electrical Engineering, KTH Royal Institute of Technology, Stockholm, Sweden (e-mail: schiessl@kth.se, hzubaidy@kth.se, skoglund@kth.se, james.gross@ee.kth.se).}
}

\date{\today}

\maketitle

\newcommand{\pow}{{\mathcal{P}}}
\newcommand{\snravg}{{\bar{\gamma}}}
\newcommand{\pout}{{\varepsilon_\mathrm{out}}}
\newcommand{\epsout}{\varepsilon'_\mathrm{out}}
\newcommand{\tr}{{\mathrm{T}}}

\newcommand{\nt}{{m}}
\newcommand{\nd}{{n}}
\newcommand{\ntotal}{{n_\mathrm{slot}}}
\newcommand{\epsupper}{\varepsilon_\mathrm{upper}}

\newcommand{\hhat}{{\hat{h}_\mathrm{scaled}}}
\newcommand{\Hhat}{{\hat{H}_\mathrm{scaled}}}
\newcommand{\Hmmse}{{\hat{H}}}
\newcommand{\hmmse}{{\hat{h}}}
\newcommand{\expected}[1]{\mathbb{E}\left[#1\right]}
\newcommand{\expectedwrt}[2]{\mathbb{E}_{#1}\left[#2\right]}
\newcommand{\Prob}[1]{\mathbb{P}\left\{#1\right\}} 

\newcommand{\dispersion}{{\mathcal{V}}}
\newcommand{\backoffsnr}{{\beta}}
\newcommand{\arrpersymb}{{\bar{\alpha}}}

\newcommand{\Xvec}{{\vec{X}}}
\newcommand{\xvec}{{\vec{x}}}
\newcommand{\Noisevec}{{\vec{N}}}
\newcommand{\Yvec}{{\vec{Y}}}

\newcommand{\noisecomb}{{N}}
\newcommand{\noisealt}{{Z}}
\newcommand{\noisealtrot}{{\bar{Z}}}
\newcommand{\corrcoff}{{\rho}}
\newcommand{\Ub}{{U_\FBL}}
\newcommand{\deltau}{{U_\delta}}
\newcommand{\Uc}{{U_\CSIFBLshort}}
\newcommand{\Uctheorem}{{\overline{U}_\CSIFBLshort}}
\newcommand{\Ue}{{U_e}}
\newcommand{\sigmaa}{{\sigma_\iCSI}}
\newcommand{\sigmaasq}{{\sigma^2_\iCSI}}
\newcommand{\sigmab}{{\sigma_\FBL}}
\newcommand{\sigmabsq}{{\sigma^2_\FBL}}
\newcommand{\sigmac}{{\sigma_\CSIFBLshort}}
\newcommand{\sigmacsq}{{\sigma^2_\CSIFBLshort}}
\newcommand{\sigmanoise}{\sigma_N}
\newcommand{\Snr}{{\Gamma}}
\newcommand{\Snrcondmeas}{{\Snr_{|\snrmmse}}}
\newcommand{\Hcondmeas}{{H_{|\hhat}}}
\newcommand{\Snrapprox}{{\widetilde{\Snr}}}
\newcommand{\snr}{{\gamma}}

\newcommand{\expgoodput}{{\overline{r}}}
\newcommand{\rate}{{r}}
\newcommand{\fraccap}{{\kappa}}
\newcommand{\Rate}{{R}}
\newcommand{\ratefin}{{r_\FBL}}
\newcommand{\ratecomb}{{r_\CSIFBLshort}}
\newcommand{\snrmmse}{{\hat{\snr}}}
\newcommand{\Snrmmse}{{\hat{\Snr}}}
\newcommand{\Snrerr}{{\tilde{\Snr}}}
\newcommand{\snrerr}{{\tilde{\snr}}}
\newcommand{\snrerrgauss}{{\snrerr_\mathrm{G}}}
\newcommand{\Snrerrgauss}{{\Snrerr_\mathrm{G}}}
\newcommand{\snrerrdelta}{{\snrerr_\delta}}
\newcommand{\Snrerrdelta}{{\Snrerr_\delta}}
\newcommand{\Zbecap}{{Z}}
\newcommand{\CDF}{{F}}
\newcommand{\Ssnr}{{\mathcal{S}}}

\newcommand{\Ssnri}{{\Ssnr_i}}
\newcommand{\Sbiti}{{S_i}}
\newcommand{\Mellin}{\mathcal{M}}
\newcommand{\Yfun}{Y}
\newcommand{\s}{{\theta}}
\newcommand{\ta}{{\tau}}
\newcommand{\tb}{{t}}

\newcommand{\fblequivlower}{{\mathrm{b,lower}}}
\newcommand{\fblequiv}{{\mathrm{b}}}
\newcommand{\FBL}{{\mathrm{FBL}}}
\newcommand{\FBLshort}{{\mathrm{F}}}
\newcommand{\iCSI}{{\mathrm{ICSI}}}
\newcommand{\iCSIshort}{{\mathrm{IC}}}
\newcommand{\perrorexact}{{ \varepsilon }}
\newcommand{\perrornorm}{{ \varepsilon }}
\newcommand{\epstarget}{{ \varepsilon' }}
\newcommand{\epsappendix}{{ \varepsilon' }}
\newcommand{\perrorpositive}{{ p_\mathrm{und} }}
\newcommand{\fblequivknown}{{  \mathrm{b},\pCSI  }}
\newcommand{\CSIFBL}{{  \mathrm{\iCSI,FBL}  }}
\newcommand{\CSIFBLshort}{{  \mathrm{IC,F}  }}

\newcommand{\Udeltab}{{U_{\delta}}}

\newcommand{\magsq}[1]{\left|#1\right|^2}
\newcommand{\defined}{{\;\overset{\Delta}{=}\;}}

\newcommand{\base}{e}
\newcommand{\Mfun}[1]{\mathcal{K}\left(#1\right)}
\newcommand{\Mfunt}[1]{\mathcal{K}\left(#1\right)}
\newcommand{\pv}{{p_\mathrm{v}}}
\newcommand{\minx}{{(min,$\times$)}}
\newcommand{\Abit}{\mathit{A}}
\newcommand{\Dbit}{\mathit{D}}
\newcommand{\Sbit}{\mathit{S}}
\newcommand{\Abitcum}{{\mathbf{A}}}
\newcommand{\Dbitcum}{{\mathbf{D}}}
\newcommand{\Sbitcum}{{\mathbf{S}}}
\newcommand{\Asnrcum}{{\boldsymbol{\mathcal{A}} }}
\newcommand{\Ssnrcum}{{\boldsymbol{\mathcal{S}} }}
\newcommand{\Ai}{{\mathit{A}_i}}
\newcommand{\Di}{{\mathit{D}_i}}
\newcommand{\Si}{{\mathit{S}_i}}
\newcommand{\Delay}{{\mathit{W}}}
\newcommand{\Asnr}{\mathcal{A}}
\newcommand{\Dsnr}{\mathcal{D}}
\newcommand{\Wsnr}{\mathcal{W}}
\newcommand{\Xsnr}{\mathcal{X}}
\newcommand{\Ysnr}{\mathcal{Y}}
\newcommand{\ptarget}{\hat{p}}
\newcommand{\snrfun}{{f}}

\newtheorem{lemma}{Lemma}
\newtheorem{theorem}{Theorem}
\newtheorem{conjecture}{Conjecture}
\newtheorem{assumption}{Assumption}
\newtheorem{corollary}{Corollary}


\begin{abstract}

With the rise of critical machine-to-machine applications, next generation wireless communication systems must be designed with strict constraints on the latency and reliability.
A key question in this context relates to channel state estimation, which allows the transmitter to adapt the code rate to the channel.
In this work, we characterize the trade-off between the estimation sequence length and data codeword length: shorter channel estimation leaves more time for the actual payload transmission but reduces the estimation accuracy and causes more decoding errors. Using lower coding rates can mitigate this effect, but may result in a higher backlog of data at the transmitter. 
We analyze this trade-off using queueing analysis on top of accurate models of the physical layer, which also account for the finite blocklength of the channel code.
Based on a novel closed-form approximation for the error probability given the rate, we show that finding the optimal rate adaptation strategy becomes a convex problem. 
The optimal rate adaptation strategy and the optimal training sequence length, which both depend on the latency and reliability constraints of the application, can improve the delay performance by an order of magnitude, compared to suboptimal strategies that do not consider those constraints.




\end{abstract}


\section{Introduction}
\label{sec:introduction}
While wireless networks have traditionally been optimized for the typical requirements of human-related services (which includes voice communication as well as Internet applications), a new class of machine-to-machine applications has been arising over the last decade.
This class can be separated into two groups.
The first group consists of so called massive machine-to-machine applications, where for example sensor readings need to be conveyed to a central data collection point. 
For such applications, the major challenges are the energy efficiency and the scalability of the communication service to potentially thousands of terminals. However, the requirements on the latency and reliability of the communication are only moderate in this case.
The second group contains critical machine-to-machine applications, which are foremost encountered in the context of industrial automation and are traditionally realized by specialized wired networks.
Due to increasing flexibility demands, and in order to enable entirely new designs of automation systems, wireless transmissions become more and more attractive.
Critical machine-to-machine applications typically generate small payloads periodically or at event-triggered time points, and require transmission with very low latency and ultra-high reliability.
For instance, automation applications from manufacturing easily require communication latencies between a sensor and a control unit below $1$ ms, as well as packet delivery ratios (with respect to that deadline) of $1 - 10^{-6}$ and above.
This area of critical machine-to-machine communications still poses significant challenges with respect to wireless network design.

Traditionally, physical layer analysis has been based on the assumption that error-free transmissions can be achieved through channel coding at Shannon's channel capacity, which is a fairly accurate model when the blocklength of the channel code is very large. However, with target latencies below $1$~ms, systems will only be able to spend a small number of symbols onto a single transmission, thus the blocklength becomes small, resulting in a significant performance loss due to channel coding at finite blocklength. 
In the finite blocklength regime, transmission errors occur due to ``above average" noise occurrences. Although the transmitter can reduce the probability of transmission errors by selecting a rate lower than the channel capacity, transmission errors are inevitable \cite{polyanskiy2010channel}. 
Understanding the implications from finite blocklength coding in combination with the fading effects of wireless communication channels is fundamental to the efficient design of ultra-reliable low latency wireless networks.
An important question in this context relates to the most efficient transmission strategy, in particular, if and how the transmitter should adapt the rate of the channel code to the instantaneous state of the wireless channel. This is still an open question in the context of low-latency communications, as rate adaptation requires time to estimate the current channel state, which reduces the already short time for data transmission. 
While spending more time on channel estimation improves the accuracy of the estimate and allows the transmitter to send the payload more reliably, this decreases the amount of time left for payload transmission even further. The shortened payload transmission duration is penalized also through a higher sensitivity to the finite blocklength effects.
This trade-off between the time spent on channel estimation and the time spent on payload transmission cannot be characterized by a purely information-theoretic analysis, as an adaptation to the channel state leads to a time-varying transmission rate that affects the backlog, and hence the latency, at the transmitter.
Therefore, in addition to the impact of finite blocklength and imperfect channel state information (CSI) on the physical layer performance, queueing effects on the link layer must be considered to address this problem.
To our best knowledge, a performance analysis that takes all these effects into account does not exist in the literature. 

As this work relates to both information theory and communication networking, we build our work on literature in both fields. 
Concerning research in information theory, the comprehensive analysis of the theoretical limits of finite blocklength channel coding by Polyanskiy et al. \cite{polyanskiy2010channel} was extended by Yang et al. to block-fading channels \cite{yang2012diversity, yang2013quasi,yang2014quasi}. 
Surprisingly, the authors found that for many types of fading channels, the maximum achievable data rate shows little dependency on the blocklength and is in fact well approximated by the outage capacity. These works even accounted for the fact that the channel state may not be known \emph{a priori}, i.e. at the start of the transmission. 
However, the data rate was assumed to be fixed and rate adaptation with imperfect CSI at the transmitter was not considered.
Analysis of imperfect CSI often focuses on the receiver side: imperfect CSI at the receiver (CSIR) will cause an error in the amplitude and phase of the signal during demodulation and decoding, which can cause errors. 
M{\'e}dard \cite{medard2000effect} investigated this by computing the mutual information of a system with imperfect CSIR.
Hassibi and Hochwald \cite{hassibi2003much} investigated the impact of imperfect CSIR on the ergodic capacity in multi-antenna systems. 
However, those information-theoretic results are based on statistical averages of the estimation error and therefore only apply when decoding is performed over infinitely many fading blocks, which would cause infinite delay. 
Moreover, rate adaptation at the transmitter cannot be analyzed using such models. 
The authors in \cite{mengyang2014constellation,yangmeng2015new} studied rate adaptation at finite blocklength with restricted input alphabets, as well as performance bounds for binary-input channels, but did not consider imperfect CSI. 
Yang et al. \cite{yang2015optimum} also studied the finite-blocklength performance when the transmitter adapts the power (but not the coding rate) to the perfectly known channel state.
Finally, Lim and Lau \cite{lim2008fundamental} and Lau et al. \cite{lau2006cross} studied rate adaptation where the transmitter has imperfect or outdated CSI, but considered neither finite blocklength effects nor the impact of transmission errors on the delay. 

In the field of communication networks, queueing theory has been used extensively to analyze the delay performance of wireless networks. While wireless network analysis poses a significant challenge to traditional queuing theory, several techniques have been developed in the last decade to address this challenge.  
Wu and Negi \cite{wu2003effective} developed the framework of effective capacity that provides approximations on the delay performance, which are however asymptotic, i.e. only valid for long delays. 
Al-Zubaidy et al. \cite{alzubaidy2016ton} used stochastic network calculus in a transform domain, which not only provides non-asymptotic bounds on the delay performance, but can also be extended for the analysis of multi-hop wireless links \cite{alzubaidy2016ton, petreska2015recursive}.
Finite blocklength effects and imperfect CSI have been separately studied with respect to their impact on the queueing performance. 
Wu and Jindal \cite{wu2011coding} analyzed a system with automatic repeat requests in fading channels at finite blocklength. 
Gursoy \cite{gursoy2013throughput} computed the effective capacity for block-fading channels at finite blocklength and showed that there is a unique optimum for the error probability. 
In our own previous work \cite{schiessl2015delay}, we extended this analysis using stochastic network calculus and provided analytical solutions for finite blocklength coding in Rayleigh block-fading channels.
Nevertheless, none of these works considered imperfect CSI. 
The queueing performance under rate adaptation with imperfect/outdated CSI but without finite blocklength effects was analyzed by Gross \cite{gross2012scheduling}.

Thus, in this work we address the delay performance of a wireless communication system in the finite blocklength regime with rate adaptation based on imperfect channel estimation. 
In this context we provide three main contributions:
\begin{itemize}
\item 
Based on stochastic network calculus \cite{alzubaidy2016ton}, we characterize the trade-off between the rate and the error probability with respect to the delay performance as a convex optimization problem. Thus, the transmitter can efficiently determine the optimal transmission rate, taking the overall latency and reliability constraints of the application flow into account.
\item This optimization is based on a novel closed-form approximation for the error probability due to the combined effects of finite blocklength coding and imperfect CSI at the transmitter in Rayleigh fading channels. Specifically, we derive an approximation for an information-theoretic result from Yang et al. \cite{yang2014quasi}. A key challenge that we overcame in this analysis is that finite blocklength effects are modeled as variations in the rate, whereas imperfect CSI corresponds to variations in the SNR. Our approximation is invertible, providing a direct mapping from the error probability to the rate, and has potential uses beyond the scope of this paper.
\item We show through numerical analysis that both finite blocklength coding and imperfect CSI have a significant impact on the performance under strict latency and reliability constraints. Our results show that rate adaption, despite needing a fraction of the resources for channel training, significantly outperforms fixed rate transmissions when low latency is required. Moreover, we find that through an optimal rate adaptation strategy and through an optimal choice of training duration versus payload transmission time, the system can improve the overall reliability by one order of magnitude.
\end{itemize}

This paper is organized as follows: The system model is given in Sec.~\ref{sec:system_model}. 
Our main contributions are presented in Sec.~\ref{sec:analysis}.
Numerical results are then presented in Sec.~\ref{sec:numerics}, followed by our conclusions in Sec.~\ref{sec:conclusions}.

Throughout the paper, we utilize the following notation: 
Uppercase italic letters $X$ generally refer to random variables, whereas the corresponding lowercase letters $x$ refer to a realization of that random variable. 
We write $f_{X|y}(x)$ for the probability density function and $F_{X|y}(x)$ for the cumulative distribution function of the variable $X$, conditioned on the value $Y=y$ of a different random variable $Y$. For complex values, $\Re\left\{x\right\}$ describes the real part, $\angle(x)$ describes the phase, and $x^*$ is the complex conjugate. 


\section{System Model}
\label{sec:system_model}

We consider data transmission over a point-to-point wireless link. 
More precisely, a data flow arrives at a transmitter and needs to be transmitted to a receiver.
The incoming data must be transmitted with probabilistic constraints on the quality-of-service: a target deadline can be violated with a probability not exceeding  
a given maximum delay violation probability.
In general, we are interested in data flows as arising in industrial contexts, i.e., with a low constant data rate, periodic arrivals, short deadlines and very low target violation probabilities.
A time-slotted system with equal slots containing $\ntotal$ symbols is considered.
Each time slot is furthermore assumed to be split into two phases: The training/estimation phase, where the transmitter sends a known training sequence of $\nt$ symbols to the receiver; and the actual data transmission phase of $\nd=\ntotal-\nt$ symbols. 
At the end of the training phase, the receiver sends the estimated CSI as feedback to the transmitter. 
Then, based on the feedback and the amount of data backlogged, the transmitter attempts to transmit a certain amount of data during the transmission phase. 
We assume that the feedback is instantaneous and error-free. 
Furthermore, the feedback also includes an acknowledgment of the previous data transmission.
Fig.~\ref{fig:quasistatic_fading} illustrates the basic system operation: since the transmitter knows only a noisy estimate of the channel capacity, it often selects a rate below the estimated capacity to decrease the chance of transmission errors, which can occur due to both imperfect channel knowledge and finite blocklength effects. 
\begin{figure}[t]
	\centering
	\def\svgwidth{\figurewidth}
	\scriptsize{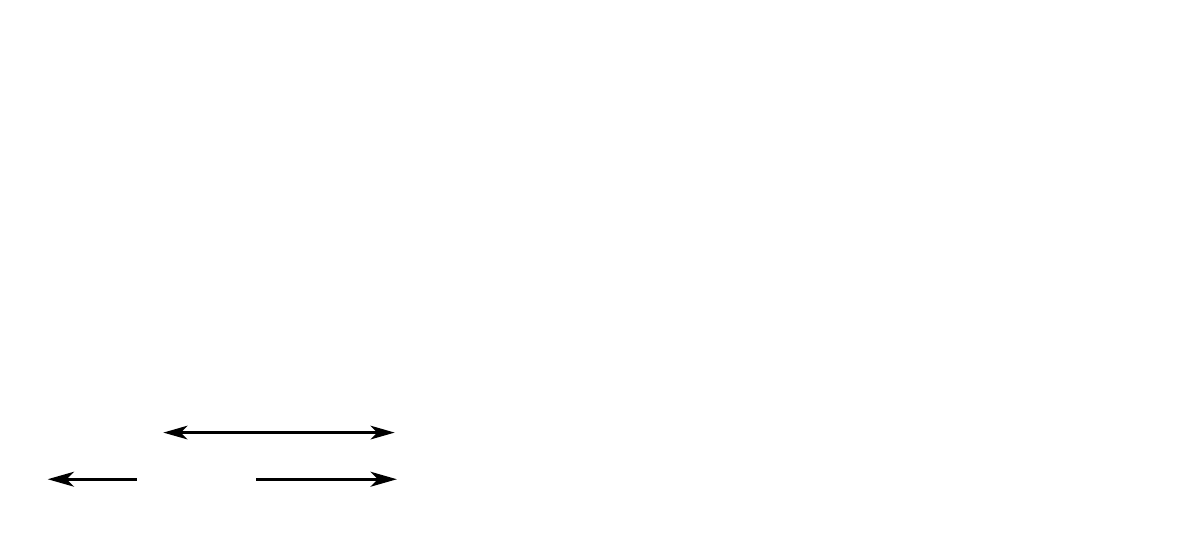}
	\caption{Wireless transmission for the first three time slots. After the training phase of $\nt$ symbols, the transmitter receives an estimate $\snrmmse$ of the channel's SNR as feedback and transmits a codeword at a rate $r(\snrmmse)$. Rate adaptation is challenging when the transmitter has only an imperfect estimate $\log_2(1+\snrmmse)$ of the channel capacity.}
	\label{fig:quasistatic_fading}
\end{figure}

Transmitter and receiver are both assumed to be static.
A frequency-flat Rayleigh block-fading channel model is assumed, where the channel remains constant for the duration of one time slot and changes independently between time slots.
This model applies for example to systems that apply frequency hopping after each time slot.
We assume that there is one transmit and one receive antenna, so the channel can be described by the scalar complex fading coefficient $H$. 
For the Rayleigh block-fading channel, $H$ has circularly symmetric Gaussian distribution $\mathcal{CN}(0,1)$. 
In each time slot, the instantaneous signal-to-noise ratio (SNR) at the receiver is given as $\Snr=\snravg|H|^2$ and has exponential distribution with mean $\snravg$.
The average SNR $\snravg$ is assumed to be constant and known at the transmitter and the receiver. 
In the following, we first provide more details on the two different system operation phases, before we relate these phases to the queuing performance and the problem statement.

\subsection{Training Phase}
\label{ssec:system_training}
The receiver estimates the fading coefficient $H$ through a known training sequence of $\nt$ symbols.
The minimum mean square error (MMSE) estimate for $H$ is given as \cite{hassibi2003much,caire2010multiuser}:
\begin{align}
\Hmmse &= \frac{\snravg \nt}{1+\snravg \nt} H + \frac{\sqrt{\snravg \nt}}{1+\snravg \nt}\noisecomb \;,
\label{eq:hmmse_def}
\end{align}
where $\noisecomb\sim\mathcal{CN}(0,1)$ is independent of $H$. 
Therefore, the channel estimate $\Hmmse$ is distributed as $\Hmmse\sim\mathcal{CN}\left(0,\rho^2\right)$ with $\rho^2=\snravg \nt/(1+\snravg \nt)$, and the estimated SNR $\Snrmmse\defined\snravg|\Hmmse|^2$ follows an exponential distribution with mean $\rho^2 \snravg$.
Due to $H$ and $\Hmmse$ being jointly Gaussian, $H$ can be expressed in terms of the estimate $\Hmmse$ as \cite{caire2010multiuser}
\begin{align}
H &= \Hmmse + \noisealt \;,
\label{eq:H_redefinition1}
\end{align}
where the estimation noise $\noisealt\sim\mathcal{CN}(0,\sigmanoise^2)$ is independent of $\Hmmse$ with
\begin{align}
\sigmanoise^2 = \frac{1}{1+\snravg\nt} \;.
\end{align}

\subsection{Data Transmission Phase}
\label{ssec:data_transmission}
After the training phase, a noise-free feedback provides the transmitter with the estimated coefficient $\Hmmse$. 
In the rest of this section we will consider a single time slot, in which the transmitter has a specific estimate $\hmmse$ of the fading coefficient.  
In this scenario, the phase of the channel coefficient is not used at the transmitter, so it is not relevant whether the transmitter knows $\hmmse$ or the estimated SNR $\snrmmse=\snravg|\hmmse|^2$.
We assume that the transmitter constantly operates at the maximum allowed transmit power, i.e., it does not perform power allocation but uses the channel estimate only to select a code rate $r$ and encode $\nd \cdot r$ data bits into a codeword of $\nd$ symbols.
However, the actual SNR $\Snr$ is unknown and can be lower than the estimated SNR $\snrmmse$, and thus the transmitter may select a code rate that is higher than the channel capacity $C=\log_2(1+\Snr)$. In this case, the channel is said to be in outage. 
When the blocklength of the channel code becomes very large (infinite), the outage probability becomes equal to the probability that the received signal cannot be correctly decoded. 
Conditioned on the estimate $\snrmmse$, the outage probability is given as
\begin{equation}
\pout = \Prob{\log_2(1+\Snr) < \rate \left| \Snrmmse = \snrmmse\right.} \;.
\label{eq:p_outage}
\end{equation} 
To compute the outage probability, first note that the distribution of $H$ conditioned on a known $\hmmse$ is 
$H = \hmmse + \noisealt$,
according to (\ref{eq:H_redefinition1}). 
Thus, when conditioned on the estimate, the channel coefficient $H$ has the same distribution as the fading coefficient of a Rician channel with line-of-sight (LoS) component $\hmmse$ (which is known at transmitter and receiver) and unknown non-LoS component $\noisealt$. 
The SNR $\Snr=\snravg|H|^2$ conditioned on the estimate $\snrmmse$ then follows a non-central $\chi^2$-distribution with two degrees of freedom and PDF:
\begin{equation}
f_{\Snr|\snrmmse}(x) = \frac{1}{\snravg \cdot \sigmanoise^2} \cdot e^{-\frac{x+\snrmmse}{\snravg \cdot \sigmanoise^2}} \cdot I_0\left( \frac{2\sqrt{x\snrmmse}}{\snravg \cdot \sigmanoise^2}\right) \;.
\label{eq:dist_snr_noncentral}
\end{equation}
The conditional outage probability can be found through the cumulative distribution, which is given by the Marcum Q-function $Q_1(a,b)$ \cite{nuttall1975some,gross2012scheduling}:
\begin{align}
 \CDF_{\Snr|\snrmmse}(x) 
 &= 1 - Q_1\left(\sqrt{\frac{2\snrmmse}{\snravg\sigmanoise^2}},\sqrt{\frac{2x}{\snravg \sigmanoise^2}}\,\right) \;.
 \label{eq:dist_snr_ccdf1}
\end{align}

In contrast, when the blocklength of the channel code is rather small (finite), the interpretation of the conditional channel distribution as that of a Rician fading channel with known LoS component $\hmmse$ enables the use of finite blocklength results for quasi-static fading channels from Yang et al. \cite{yang2013quasi,yang2014quasi}. 
According to the normal approximation \cite[(59-61)]{yang2014quasi}, the approximation $\perrornorm$ for the error probability for a code with rate $r$ is given by:
\begin{align}
\perrornorm &\approx \expected{\left.Q\left(\frac{\log_2(1+\Snr) - \rate}{\sqrt{\dispersion(\Snr)/\nd}}\right) \right| \Snrmmse = \snrmmse} \;, 
\label{eq:yang_normalapprox}
\end{align}
where $Q(\cdot)$ is the Gaussian Q-function and the channel dispersion is given as:
\begin{equation}
\dispersion(\snr)=\log_2^2(e) \left(1-\frac{1}{(1+\snr)^2}\right) \;.
\label{eq:dispersion_polyanskiy}
\end{equation}
The expectation is taken over the conditional distribution (\ref{eq:dist_snr_noncentral}) of the SNR $\Snr$ given the estimated SNR $\snrmmse$. 
We finally state two remarks concerning (\ref{eq:yang_normalapprox}): 

(I) The authors in \cite{yang2014quasi} assumed perfect CSI at the receiver (CSIR) when proposing (\ref{eq:yang_normalapprox}).
In our scenario, the receiver does not have perfect knowledge of $H$, as it only knows the estimate $\hmmse$ but not the estimation noise $\noisealt$. 
We can verify the accuracy of (\ref{eq:yang_normalapprox}) by numerically computing a lower bound \cite[Cor.~3]{yang2013quasi} on the achievable rate for Rician fading channels that does not depend on perfect CSIR\footnote{The achievable rate in \cite[Cor.~3]{yang2013quasi} is for fading channels with no instantaneous CSI, which means the realization $\noisealt$ is unknown at both transmitter and receiver, but the statistics (i.e., the mean $\hmmse$) are known. This corresponds exactly to our model, and can be computed with the SPECTRE toolbox \cite{polyanskiy2016spectre}.}. 
These results indicate that, in our scenario with a significant amount of training and rate adaptation at the transmitter, (\ref{eq:yang_normalapprox}) provides a reasonable approximation for the achievable performance.
A second and independent argument to support the accuracy of (\ref{eq:yang_normalapprox}) is that after the training and feedback phase, many wireless communication systems send pilot symbols in addition to the codeword symbols in the data transmission phase, which does not change the CSI at the transmitter but significantly improves the CSI at the receiver, such that imperfect CSI at the receiver will not affect the decoding performance. 

(II) The finite blocklength result is an approximation and not accurate for all parameters. 
For the non-fading AWGN channel \cite{polyanskiy2010channel}, the approximation is considered to be accurate when the blocklength $n$ is above 200. 
Furthermore, it may become less accurate at extremely low error probabilities, e.g., $\perrornorm<10^{-4}$. 
In this work, we will avoid very short blocklengths and very low error probabilities, as ultra high reliability with respect to a certain target deadline is achieved through retransmissions. 
Thus, we will assume that the approximation of $\perrornorm$ in (\ref{eq:yang_normalapprox}) is exactly equal to the error probability.

\subsection{Queueing Model}
In the wireless channel model, transmissions fail with probability $\perrornorm$ due to imperfections in the channel state estimation and due to finite blocklength effects. 
Furthermore, as the data rate $\rate$ is adapted to the time-varying channel, it varies from time slot to time slot and may become very small.
Thus, not all the data arriving in a certain time slot can be transmitted successfully. 
To avoid data loss, data must be stored in a buffer or queue, in which it will remain for a random time until the receiver sends an acknowledgment indicating that the data was successfully decoded. 
As the arrival rate of the considered data flow is rather small, the buffer is assumed to be large enough to hold all incoming data, such that all data will eventually be transmitted.

One performance metric for such a transmission scheme is the expected goodput, which is the expected value of the amount of data that can be successfully transmitted. 
In each time slot, $r \cdot\nd$ bits can be transmitted, but transmissions fail with probability $\perrornorm$. 
We define the normalized expected goodput (normalized with respect to the total number of symbols $\ntotal=n+m$) as:
\begin{equation}
 \expgoodput \defined \frac{\nd}{\ntotal} \expected{\rate(\Snrmmse)\cdot\left(1-\perrornorm\right)} \;,
 \label{eq:expected_goodput1}
\end{equation}
where the expectation is taken over the distribution of the estimated SNR $\Snrmmse$, which follows an exponential distribution according to (\ref{eq:hmmse_def}). We specifically denote the rate as $r(\Snrmmse)$ as it is chosen according to the estimated SNR $\Snrmmse$.
	
However, for the considered data flow (originating from an industrial application), the expected goodput is not a suitable performance metric as it cannot characterize the tail of the random delay distribution of the system.
Thus, in order to analyze the random delay of data in the queue, the system is described in terms of its arrival, service and departure processes. 
In time slot $i$, the arrival process $\Ai$ is given by the data that is generated at the transmitter side, e.g., from sensor readings, and enters the queue. 
The service process $\Si$ describes the number of bits that can potentially be transmitted over the wireless channel in time slot $i$. 
In case of transmission errors, i.e., with probability $\perrornorm_i$, the service $\Si$ is zero: 
\begin{align}
	\Si=\left\{\begin{array}{ll}
		n\cdot r(\Snrmmse_i) &\quad \text{with prob.} \quad (1-\perrornorm_i) \\
		0 &\quad \text{with prob.} \quad \perrornorm_i
	\end{array}\right.
	\label{eq:def_service}
	\; ,
\end{align}
where the estimated SNR $\Snrmmse_i$, as well as the corresponding $\rate(\Snrmmse_i)$ and $\perrornorm_i$ vary from one time slot to another.
The departure process $\Di$ is given as the number of bits leaving the queue, which is equal to the number of bits that reach the receiver successfully. The departure process is upper-bounded both by the service process and by the amount of data waiting in the queue. 
For the analysis of queueing systems, we also define the cumulative arrival, service, and departure processes
\begin{equation}
	\Abitcum(\ta,\tb) \defined \sum\limits_{i=\ta}^{\tb-1}\Ai \;,
	\quad\Sbitcum(\ta,\tb) \defined \sum\limits_{i=\ta}^{\tb-1}\Si \;,
	\quad\Dbitcum(\ta,\tb) \defined \sum\limits_{i=\ta}^{\tb-1}\Di \;.
\end{equation}

The random delay $\Delay(t)$ at time $t$ is then defined as the time it takes before all data that arrived before time $t$ has actually departed from the queue, i.e. reached the receiver:
\begin{equation}
	\Delay(t) \defined \inf\left\{u\geq 0:\quad \Abitcum(0,\tb) \leq \Dbitcum(0,\tb+u) \right\} \;.
	\label{eq:def_delay}
\end{equation}
For the considered application, we are interested in conveying the data with respect to a certain deadline.
This translates into the delay violation probability as a performance metric, i.e. the maximum probability at any time $t$ that the random delay $\Delay(t)$ exceeds a specified target delay $w$ (i.e., deadline):
\begin{equation}
	\pv(w) \defined \sup_{t\ge 0} \left\{ \Prob{\Delay(t)>w} \right\}\; .
	\label{eq:def_pdelayviol_time}
\end{equation}

\subsection{Problem Statement}
The main objective of this work is a characterization of the trade-off between the time spent on channel training and time spent on actual data transmission. 
When using a long training sequence of $\nt$ symbols, the channel estimates become more accurate, allowing transmissions with higher reliability and hence fewer retransmission attempts.
However, a long training sequence reduces the number of remaining symbols for data transmission $\nd=\ntotal-\nt$.
We thus face a typical trade-off between increasing the transmission reliability and the duration for payload transmissions.
This trade-off becomes in particular interesting when shorter transmission slots are considered, which is the case in low latency wireless networks for critical industrial applications: Due to finite blocklength effects, the shortening of the payload transmission deteriorates the communication performance even more rapidly.
What is then the optimal trade-off between $\nt$ and $\nd$? 
How is this optimal trade-off related to the required delay violation probability $\pv(w)$ of the system, i.e. the criticality of the conveyed application data? 
Does it even make sense to use training in all scenarios, or is it sometimes better to skip training ($\nt=0$) and use the entire time slot for data transmissions at fixed rate?

The same kind of trade-off can be observed between the code rate $\rate$ and the corresponding error probability $\perrornorm$: Lower data rates lead to lower error probability $\perrornorm$ but also to reduced data throughput. 
Therefore, our secondary objective is to determine the optimal trade-off with respect to the delay performance between code rate $\rate$ and error probability $\perrornorm$. 


\section{Analysis}
\label{sec:analysis}
In this section, we present our main contribution:
based on stochastic network calculus \cite{alzubaidy2016ton}, we formulate the optimal trade-off with respect to the delay violation probability $\pv(w)$ between the code rate $r$ and the error probability $\perrornorm$ as a convex optimization problem.
Thus, depending on the requirements of the application, an optimal rate adaptation strategy for the transmitter can be quickly determined.
The biggest challenge in the characterization of that trade-off is the fact that no closed-form representation of the problem exists in the literature, as $\perrornorm$ is not given in closed form.
Thus, after summarizing stochastic network calculus and presenting the problem in Sec.~\ref{ssec:queueing}, we derive in Sec.~\ref{ssec:impcsi} and \ref{ssec:finite_blocklength}  a novel closed-form approximation of the error probability $\perrornorm$.
This analysis initially considers only imperfect CSI at the transmitter and yields an approximation and upper bound for the outage probability $\pout$, while ignoring the effects of finite length coding. 
The approximation is extended to channel codes with finite blocklength in Sec.~\ref{ssec:finite_blocklength}. 
This allows us then in Sec.~\ref{ssec:optimal_ra} to address the optimal trade-off between the rate $\rate$ and error probability $\perrornorm$ as convex problem based on the previously developed approximations. 

In addition, we show in Sec.~\ref{ssec:further_uses} how our approximation could be applied to scenarios different from the one considered in this work. 


\subsection{Queueing Analysis}
\label{ssec:queueing}
This subsection summarizes previous results on how the delay performance, specifically the delay violation probability $\pv(w)$ given in (\ref{eq:def_pdelayviol_time}), can be analyzed through stochastic network calculus \cite{fidler2006network,alzubaidy2016ton}.
While the delay violation probability can also be analyzed using effective capacity \cite{wu2003effective}, which has been successfully applied in numerous works, e.g., \cite{gursoy2013throughput,ozmen2013throughput}, effective capacity only provides an approximation of $\pv(w)$ in the tail of the delay distribution, i.e., for relatively large delays. 
Contrary to that, stochastic network calculus \cite{fidler2006network}, which has been recently extended to wireless network analysis in a transform domain \cite{alzubaidy2016ton} and which has also been applied in various scenarios \cite{petreska2015recursive,schiessl2015delay,schiessl2016interference,al2015performance}, provides a strict upper bound on the delay violation probability $\pv(w)$, even at low delay.
This is beneficial for ultra-reliable low latency systems in an industrial context: The modeled system will violate certain delay bounds with an even lower probability than shown by the analysis.
Parts of the following summary are taken from our previous work in \cite{schiessl2015delay}. 

The delay $W(t)$ in (\ref{eq:def_delay}) is defined in terms of the arrival and departure processes. 
However, for finding the statistical distribution of the delay, it is easier to use only the arrival and service processes.
The authors in \cite{alzubaidy2016ton} characterized these processes in the exponential domain, also referred to as \emph{SNR domain}. The main benefit of this approach is the elimination of the logarithm in the channel capacity. Instead of describing the cumulative service and arrival $\Sbitcum(\ta,\tb)$ and $\Abitcum(\ta,\tb)$ in the bit domain, they are converted to the SNR domain (denoted by calligraphic letters) as follows:
\begin{equation}
\Asnrcum(\ta,\tb) \defined \base^{\Abitcum(\ta,\tb)},\quad \Ssnrcum(\ta,\tb) \defined \base^{\Sbitcum(\ta,\tb)}
\;.
\label{eq:Asnr_Ssnr}
\end{equation}
Similarly, we define $\Asnr_i \defined e^{\Abit_i}$ and $\Ssnr_i \defined e^{\Sbit_i}$. According to the system model, the service process $\Sbit_i$ is independent and identically distributed (i.i.d.) between time slots. We also require the arrival process $\Abit_i$ to be i.i.d. in this work. 

An upper bound on the delay violation probability $\pv(w)$ can then be computed in terms of the Mellin transforms of $\Asnr_i$ and $\Ssnr_i$, where we can drop the subscript $i$ due to the i.i.d. assumption.
The Mellin transform $\Mellin_\Xsnr(\s)$ of a nonnegative random variable $\Xsnr$ is defined as \cite{alzubaidy2016ton}
\begin{equation}
\Mellin_\Xsnr(\s)\defined\mathbb{E}\left[\Xsnr^{\s-1}\right]
\end{equation}
for a parameter $\s \in \mathbb{R}$. 
For the analysis, we always choose $\s>0$ and first check whether the stability condition $\Mellin_{\Asnr}(1+\s)\Mellin_{\Ssnr}(1-\s) < 1$ holds.
If it holds, define the kernel \cite{alzubaidy2016ton,schiessl2015delay}
\begin{align}
\Mfun{\s,w} &\defined \lim\limits_{t\to\infty} \sum_{u=0}^{t} \Mellin_{\Asnr}(1+\s)^{t-u} \cdot \Mellin_{\Ssnr}(1-\s)^{t+w-u}\\
&= \frac{\Mellin_{\Ssnr}(1-\s)^{w}}{1-\Mellin_{\Asnr}(1+\s)\Mellin_{\Ssnr}(1-\s)} 
\;.
\label{eq:snc_kernel}
\end{align}
This kernel is an upper bound for the delay violation probability, which holds for any time slot $t$, including the limit $t\to\infty$ (steady-state):
\begin{equation}
\pv(w) \leq \inf_{\s>0}\left\{ \Mfun{\s,w} \right\} \;.
\label{eq:pdelay_violation}
\end{equation}
For any parameter $\s>0$, the kernel $\Mfun{\s,w}$ provides an upper bound on the probability $\pv(w)$ that the delay exceeds the target delay $w$. 
In order to find the tightest upper bound, one must find the parameter $\s>0$ that minimizes $\Mfun{\s,w}$.

The kernel $\Mfun{\s,w}$ depends on the Mellin transforms of $\Asnr$ and $\Ssnr$, i.e., of the arrival and service processes in the SNR domain. 
For simplicity, we assume that the arrival process is constant, i.e., in each time slot of length $\ntotal$ symbols, a data packet with a constant size of $\arrpersymb\cdot\ntotal$ bits arrives at the transmitter, thus $\Mellin_\Asnr(\s)=e^{\arrpersymb\ntotal(\s-1)}$. 
The service process describes the number of bits that are \emph{successfully} transmitted to the receiver. 
The random service $\Sbit$ in (\ref{eq:def_service}) can be described as $\Sbit=\nd\rate(\Snrmmse) \cdot Z$, where $\rate(\Snrmmse)$ is the code rate adapted to the measured SNR $\Snrmmse$ and $Z$ is a Bernoulli random variable, which is zero in case of error, i.e., with probability $\perrornorm$, and one in case of successful transmission, i.e. with probability $(1-\perrornorm)$. 
Thus, the Mellin transform $\Mellin_{\Ssnr}(\s)$ of the service process in the SNR domain $\Ssnr=e^\Sbit$ is given as \cite{schiessl2015delay}
\begin{align}
\Mellin_\Ssnr(\s) &= \expected{\Ssnr^{\s-1}}=\expectedwrt{\Snrmmse,Z}{e^{\nd\rate(\Snrmmse)\cdot Z\cdot(\s-1)}}
\label{eq:mellin_service}
\\
&= \int\limits_{0}^{\infty}\left((1-\perrornorm)e^{
	\nd\rate(\snrmmse)(\s-1)
} + \perrornorm\right)f_\Snrmmse(\snrmmse)d\snrmmse
\;.
\label{eq:mellin_service_comb1}
\end{align}
Note that, in general, the error probability $\perrornorm$ is not constant but varies based on the estimated SNR $\snrmmse$ and on the selected code rate $r(\snrmmse)$.
When $\Mellin_{\Asnr}(\s)$ and $\Mellin_{\Ssnr}(\s)$ are known\footnote{In order to evaluate $\Mellin_\Ssnr(\s)$ numerically, the integration range of (\ref{eq:mellin_service_comb1}) can be split into a finite set of intervals. As both the error probability and $f_\Snrmmse(\snrmmse)$ decrease in $\snrmmse$, replacing $\snrmmse$ in each interval with the lower limit of that interval will result in an upper bound on $\Mellin_\Ssnr(\s)$. The delay bound (\ref{eq:pdelay_violation}) remains valid when $\Mellin_{\Ssnr}(\s)$ is replaced by its upper bound.}, the kernel $\Mfun{\s,w}$ and the upper bound (\ref{eq:pdelay_violation}) on the delay violation probability $\pv(w)$ can be computed.

In this work, we seek to characterize the optimal trade-off with respect to the delay violation probability $\pv(w)$ between training length $\nt$ and code length $\nd$, as well as between the selected rate $r(\snrmmse)$ and the resulting error probability $\perrornorm$. 
The analytical bound (\ref{eq:pdelay_violation}) on $\pv(w)$ can be used to solve this problem. 
First, in order to obtain the delay bound (\ref{eq:pdelay_violation}), one must iterate over different parameters $\s>0$, and compute the kernel $\Mfun{\s,w}$ for each $\s$. 
The kernel $\Mfun{\s,w}$ is monotonically increasing in $\Mellin_\Ssnr(1-\s)$. 
Therefore, for each particular value of $\s>0$, finding the parameters $\nt$ and $r(\snrmmse)$ (with the corresponding values of $\nd=\ntotal-\nt$ and $\perrornorm$) that minimize $\Mellin_\Ssnr(1-\s)$ in (\ref{eq:mellin_service_comb1}) yields the desired minimum on the delay bound.

Concerning the optimal training length $\nt$, one can simply iterate over all possible choices of $\nt$ and pick the value that provides the lowest value of $\Mellin_\Ssnr(1-\s)$.
However, finding the optimal rates $r(\snrmmse)$ that minimize (\ref{eq:mellin_service_comb1}) is hard, because the error probability $\perrornorm$ in (\ref{eq:yang_normalapprox}) can only be evaluated numerically and depends on the training length $\nt$, on the estimated SNR $\snrmmse$, and on the rate $\rate(\snrmmse)$. 
In order to solve this problem, we develop a closed-form approximation for $\perrornorm$ in the following two sections. 
We then show in Sec.~\ref{ssec:optimal_ra} that with this closed-form approximation, finding the rates $r(\snrmmse)$ that minimize \eqref{eq:mellin_service_comb1} becomes a convex optimization problem.
	
\subsection{Outage Probability Approximation for Imperfect CSI}
\label{ssec:impcsi}
When the blocklength of the code tends to infinity, i.e., when the effects of channel coding at finite blocklength are ignored, then the error probability $\perrornorm$ in (\ref{eq:yang_normalapprox}) converges to the outage probability $\pout$ in (\ref{eq:p_outage}). Conditioned on the channel estimate, the outage probability can be computed using (\ref{eq:dist_snr_ccdf1}), i.e., in terms of the Marcum Q-function.
In this section, we provide an upper bound for the outage probability based on the Gaussian Q-function.

\begin{lemma}{
\label{lemma_impcsi}
Given an imperfect estimate of the channel $\snrmmse$ and a rate $\rate$, the outage probability is bounded by
\begin{align}
\pout &\leq Q\left(\frac{\snrmmse-(2^\rate-1)}{\sigmaa}\right) \;,
\label{eq:outage_probability}
\end{align}
with $\sigmaasq\defined 2\sigmanoise^2\snravg\snrmmse$.
}
\end{lemma}
\begin{proof}
Given a known measurement $\hmmse$, the random variable $H$ is given in terms of the random variable $\noisealt$ according to (\ref{eq:H_redefinition1}). Thus:
\begin{align}
\Snr &= \snravg|H|^2 = \snravg(\hmmse + \noisealt)(\hmmse^* + \noisealt^*)\\
&=\snrmmse + 2\snravg\Re\left\{\hmmse^*\noisealt\right\}+ \snravg\magsq{\noisealt} \\
&=\snrmmse  + 2\snravg|\hmmse|\Re\left\{\noisealtrot\right\}+ \snravg\magsq{\noisealtrot}
\end{align}
where $\noisealtrot=e^{-j\angle(\hmmse)}\noisealt$ is just a phase-rotated version of $\noisealt$. The distribution and the magnitude of a circularly symmetric random variable stay constant under phase rotation, and thus the real part $\Re\left\{\noisealtrot\right\}$ has Gaussian distribution $\mathcal{N}(0,\sigmanoise^2/2)$.
It follows that
the SNR $\Snr=\snravg\magsq{H}$ is given as
\begin{align}
\Snr 
&= \snrmmse + \Snrerrgauss + \Snrerrdelta =\snrmmse + \Snrerr \;,
\label{eq:bound_snr_impcsi}
\end{align}
i.e. the estimation error $\Snrerr = \Snr-\snrmmse$ is the sum of a Gaussian error $\Snrerrgauss \sim\mathcal{N}(0,\sigmaasq)$ and some $\Snrerrdelta=\snravg\magsq{\noisealtrot}$. 
The outage probability $\pout$ can then be bounded as:
\begin{align}
\pout &= \Prob{\Snr < 2^\rate-1 \left| \Snrmmse=\snrmmse\right.} \\
&\leq \Prob{\snrmmse + \Snrerrgauss < 2^\rate-1}\\
&=\Prob{-\frac{\Snrerrgauss}{\sigmaa}>\frac{\snrmmse-(2^\rate-1)}{\sigmaa}}
\end{align}
where the inequality holds because $\Snrerrdelta\geq 0$. 
\end{proof}
We observe that the variance of $\Snrerrdelta$ is proportional to $\sigmanoise^4$ and thus $\Snrerrdelta$ becomes very small relative to $\Snrerrgauss$ as the channel estimates become more accurate. 
In that case, (\ref{eq:outage_probability}) becomes a tight upper bound on the outage probability.

Given a target outage probability of e.g. $\epsout=10^{-3}$, it is difficult to find the exact rate $\rate$ for which the outage probability $\pout$ is exactly $\epsout$, as the Marcum Q-function cannot be easily inverted. 
However, the approximation (\ref{eq:outage_probability}) is invertible:
\begin{corollary}{
		\label{corollary_impcsi_rate}
		Given an imperfect channel estimate $\snrmmse$ and a target outage probability $\epsout$ with $Q(\snrmmse/\sigmaa)<\epsout<1/2$, the actual outage probability $\pout$ is less than or equal $\epsout$ if the transmitter chooses the rate
		\begin{align}
		\rate_\iCSI(\snrmmse,\epsout) = \log_2\left(1+\snrmmse - \sigmaa Q^{-1}(\epsout)\right) \;.
		\label{eq:outage_rate_approx}
		\end{align}
	}
\end{corollary}
\begin{proof}
Let the right-hand side of (\ref{eq:outage_probability}) be equal to $\epsout$ and solve for the rate $\rate$. The condition $\epsout>Q(\snrmmse/\sigmaa)$ ensures that the rate is positive.
\end{proof}
Therefore, when the transmitter selects the rate $\rate_\iCSI(\snrmmse,\epsout)$, the actual outage probability $\pout$ is not exactly equal to the target outage probability $\epsout$, but $\pout$ is smaller than $\epsout$. This upper bound on the outage probability allows for a worst-case performance analysis.

\subsection{Combined Analysis of Imperfect CSI and Finite Blocklength}
\label{ssec:finite_blocklength}
When analyzing the physical layer using the finite blocklength model, we focus first on the case where the channel state information is perfect. 
If, in a specific time slot, the fading coefficient $h$ and the SNR $\snr=\snravg|h|^2$ are perfectly known at the transmitter and receiver, then (\ref{eq:yang_normalapprox}) can be computed easily, as the expected value is taken with respect to the constant $h$. In that case, (\ref{eq:yang_normalapprox}) can be solved for the achievable rate $\rate$ given the error probability $\perrornorm$: 
\begin{equation}
\ratefin(\snr, \nd, \perrorexact) \approx \log_2(1+\snr) - \sqrt{\frac{\dispersion(\snr)}{\nd}}Q^{-1}(\perrorexact) \; .
\label{eq:polyanskiy1}
\end{equation}
This result corresponds to the approximation for the achievable rate in AWGN channels by Polyanskiy et al. \cite[Thm.~54]{polyanskiy2010channel}.
We can obtain the same result from a different interpretation: For a fixed capacity $c=\log_2(1+\snr)$, define the \emph{random blocklength-equivalent capacity} 
\begin{equation}
C_\fblequiv(\snr) \defined \log_2(1+\snr) -  \sqrt{\frac{\dispersion(\snr)}{\nd}} \cdot \Ub
\label{eq:blocklength_equiv_cap}
\end{equation}
with $\Ub \sim\mathcal{N}(0,1)$ and assume that errors occur if and only if an outage occurs, i.e. iff $C_\fblequiv(\snr) < \ratefin(\snr, \nd,\perrorexact)$. This means that $\ratefin$ is interpreted as the outage capacity for a channel with random capacity $C_\fblequiv(\snr)$, which simplifies the comparison and combination of imperfect CSI and finite blocklength effects.
However, while finite blocklength effects are modeled as Gaussian variation $\Ub$ in the capacity, we observed in Sec.~\ref{ssec:impcsi} that channel estimation errors can be approximated by Gaussian variations in the SNR. In order to analyze the combined impact of both effects, we approximate the finite blocklength variations $\Ub$ as variation in the SNR. This is done using the first-order Taylor approximation of $\ln(x)$ around the point $x_0$, which has gradient $\frac{1}{x_0}$. Due to the concavity of the $\ln$-function, this linear approximation is larger than the function itself:
\begin{equation}
\ln(x_0) - \frac{1}{x_0} (x_0a) \geq \ln\left(x_0 - x_0 a\right) \; .
\end{equation}
Due to $\ln(x)$ being continuous and monotonically increasing, this means that for some $\delta\geq0$ and $b=a\log_2(e)$:
\begin{equation}
\log_2(x_0) - b = \log_2\left(x_0 - x_0 \frac{b}{\log_2(e)} + \delta\right) \;.
\end{equation}
By applying this result to (\ref{eq:blocklength_equiv_cap}) around $x_0=1+\snr$, (\ref{eq:blocklength_equiv_cap}) can be rewritten as 
\begin{equation}
C_\fblequiv(\snr) = \log_2
\left(1+\snr - \sigmab(\snr)\cdot\Ub + \deltau
\right)
\label{eq:blocklength_equiv_cap_redefined2}
\end{equation}
with
\begin{equation}
\sigmab(\snr) \defined \frac{1+\snr}{\log_2(e)}\sqrt{\frac{\dispersion(\snr)}{\nd}} \; ,
\end{equation}
 $\Ub \sim\mathcal{N}(0,1)$ and some random $\deltau\geq 0$. Thus, we convert the Gaussian error in the rate to a Gaussian error in the SNR, plus an unknown $\deltau$ which can later be ignored because it is non-negative.

As a next step, imperfect CSI is taken into account. When the transmitter has imperfect channel state information, the error probability is given by $\perrornorm$ defined in (\ref{eq:yang_normalapprox}).
The following lemma allows an easier notation and interpretation of our results.
\begin{lemma}
The approximate error probability $\perrornorm$ by Yang et al. \cite{yang2014quasi} given in (\ref{eq:yang_normalapprox}) is equal to the \emph{blocklength-equivalent outage probability}, i.e.
\begin{equation}
\expected{\left.Q\left(\frac{\log_2(1+\Snr) - \rate}{\sqrt{\dispersion(\Snr)/\nd}}\right) \right| \Snrmmse = \snrmmse} = \Prob{C_b(\Snr)<r \left| \Snrmmse = \snrmmse\right.} \;,
\end{equation}
with $C_\fblequiv$ defined by (\ref{eq:blocklength_equiv_cap}) now depending on the random SNR $\Snr$ (conditioned on $\snrmmse$) and on the random variable $\Ub$.
\end{lemma}
\begin{proof}
For a fixed SNR $\Snr=\snr$, the $Q$-function on the inside of the expectation in (\ref{eq:yang_normalapprox}) is equal to $\Prob{C_\fblequiv(\snr) < r}$ by definition of $C_\fblequiv(\snr)$. The claim follows by taking the expectation over the distribution of $\Snr$ (conditioned on the measurement $\snrmmse$) on both sides.
\end{proof}

When the SNR is not perfectly known at the transmitter, the fixed value $\snr$ needs to be replaced by the random $\Snr=\snrmmse+\Snrerr$. We have seen before that the estimation error $\Snrerr$ in the SNR $\Snr$ can also be approximated as a Gaussian error: $\Snrerr=\Snrerrgauss+\Snrerrdelta$ with $\Snrerrdelta\geq 0$. Thus, (\ref{eq:blocklength_equiv_cap_redefined2}) becomes
\begin{align}
C_\fblequiv(\Snr) &= \log_2\left(
1+\snrmmse + \Snrerr - \sigmab(\snrmmse+\Snrerr)\Ub + \deltau \right)
\\&\geq \log_2\left(
1+\snrmmse + \Snrerrgauss - \sigmab(\snrmmse+\Snrerr)\Ub \right) \; .
\label{eq:blocklength_equiv_cap_redefined3}
\end{align}
When defining the right side of (\ref{eq:blocklength_equiv_cap_redefined3}) as $C_\fblequivlower(\Snr)$, the error probability $\perrornorm$ can be bounded as
\begin{align}
\perrornorm\leq \Prob{C_\fblequivlower(\Snr) < r\left| \Snrmmse = \snrmmse\right.} \;.
\label{eq:errbound1}
\end{align}
Naturally, the channel measurement error $\Snrerr$ and its Gaussian approximation $\Snrerrgauss$ do not depend on the noise in the data transmission phase. In addition, we assumed that the decoding performance is not affected by imperfect CSI at the receiver. Therefore, $\Snrerrgauss$ and $\Ub$ are considered to be independent variables. 
To simplify the analysis, we make the following assumption:
\begin{assumption} 
\label{assumption:sigmafbl} Inequality
(\ref{eq:errbound1}) holds when $\sigmab(\snrmmse+\Snrerr)$ is replaced by $\sigmab(\snrmmse)$, i.e. 
	\begin{equation}
	\perrornorm \leq \Prob{\log_2\left(
	1+\snrmmse+\Snrerrgauss - \sigmab(\snrmmse)\Ub \right) < r}
	\label{eq:blocklength_equiv_cap_redefined4}
	\end{equation}
	is assumed to hold for all parameters.
\end{assumption}
Motivation: 
When the estimated SNR $\snrmmse$ is larger than the actual SNR, then the variance is replaced by a larger term, i.e., the finite blocklength effects are overestimated. Overestimating the negative effects of finite blocklength coding should generally lead to an overestimation of the error probability. 
On the other hand, when the estimated SNR $\snrmmse$ is smaller than the actual SNR, then the channel is already better than predicted, and there is a high margin between the actual capacity and the rate, so errors in this regime are very rare. 

\begin{lemma}
	\label{lemma_csifbl}
    If Assumption~\ref{assumption:sigmafbl} holds, and if the estimated SNR $\snrmmse$ and the average SNR $\snravg$ are known, then the error probability $\perrornorm$ for a code with rate $\rate$ is bounded as
	\begin{align}
	\perrornorm &\leq Q\left(\frac{\snrmmse - \left(2^\rate-1\right)}{\sigmac(\snrmmse)}\right) 
	\defined \epstarget \;,
	\label{eq:lemma_csifbl_prob}
	\end{align}
	with
	\begin{align}
	\sigmacsq(\snrmmse) &= \sigmaasq + \sigmabsq(\snrmmse) \;.
	\label{eq:lemma_csifbl_sigmac}
	\end{align}
\end{lemma}
\begin{proof}
	The random variables $\Snrerrgauss\sim\mathcal{N}(0,\sigmaasq)$ and $\Ub\sim\mathcal{N}(0,1)$ are independent. Thus, the difference $\Snrerrgauss -\sigmab(\snrmmse)\Ub$ can be described by $\Uc\sim\mathcal{N}(0,\sigmacsq(\snrmmse))$, where $\sigmacsq(\snrmmse)$ is the sum of the two variances. Then, starting from (\ref{eq:blocklength_equiv_cap_redefined4}) we obtain:
	\begin{align}
	\perrornorm &\leq \Prob{\snrmmse+\Snrerrgauss - \sigmab(\snrmmse)\Ub <2^r-1}
	\nonumber\\
	&= \Prob{\snrmmse + \Uc < 2^\rate-1}\nonumber\\
	&= \Prob{-\frac{\Uc}{\sigmac(\snrmmse)} > \frac{\snrmmse - \left(2^\rate-1\right)}{\sigmac(\snrmmse)}} \;. \nonumber
	\end{align}
\end{proof}
While all our numerical results confirmed that Assumption~\ref{assumption:sigmafbl} holds, we are mainly interested in an approximation of the error probability $\perrornorm$. The expression in (\ref{eq:lemma_csifbl_prob}) still provides an approximation for $\perrornorm$ even if Assumption~\ref{assumption:sigmafbl} does not hold for some parameters.
However, an upper bound for $\perrornorm$ can be very useful in the context of ultra-reliable low latency systems with rate adaptation, specifically, when the transmitter wants to select a rate such that the error probability $\perrornorm$ is below a target error probability $\epstarget$:
\begin{corollary}
	\label{corollary_csifbl_rate}
Given an imperfect channel estimate $\snrmmse$ and a target error probability $\epstarget$ with $Q(\snrmmse/\sigmac(\snrmmse))<\epstarget<1/2$, the actual error probability $\perrornorm$ is less than or equal $\epstarget$ if Assumption~\ref{assumption:sigmafbl} holds and
if the transmitter chooses the rate
	\begin{align}
	\ratecomb(\snrmmse,\nd,\epstarget) &= \log_2\left(1+\snrmmse -\sigmac(\snrmmse) 
	Q^{-1}(\epstarget)\right)
	\;.
	\label{eq:lemma_csifbl_rate}
	\end{align}
\end{corollary}
\begin{proof}
The proof follows by solving (\ref{eq:lemma_csifbl_prob}) for $r$, with $\epstarget>Q(\snrmmse/\sigmac(\snrmmse))$ ensuring that $r>0$.
\end{proof}

\subsection{Optimal Rate Adaptation}
\label{ssec:optimal_ra}
The problem addressed in this paper is finding the parameters of training sequence length $\nt$ and rates $r(\snrmmse)$ (with corresponding values of $\nd$ and $\perrornorm$) that minimize the upper bound (\ref{eq:pdelay_violation}) on the delay violation probability $\pv(w)$. 
For the rate adaptation, we will show in this section that Lemma~\ref{lemma_csifbl} and Corollary~\ref{corollary_csifbl_rate} can be used to determine a nearly optimal solution.

As we already observed in Sec.~\ref{ssec:queueing}, an optimal rate adaptation strategy minimizes $\Mellin_\Ssnr(1-\s)$ for some parameter $\s>0$, or equivalently, minimizes $\Mellin_\Ssnr(\s)$ for some $\s<1$.
An approximate solution can be found with Lemma~\ref{lemma_csifbl}, which provides an upper bound $\epstarget$ on the error probability $\perrornorm$, resulting in an upper bound on $\Mellin_\Ssnr(\s)$.
However, instead of choosing the rates $r(\snrmmse)$ and then bounding the error probability, Corollary~\ref{corollary_csifbl_rate} allows choosing a target error probability $\epstarget$ for each $\snrmmse$ and then computing the achievable rate $\ratecomb(\snrmmse,\nd,\epstarget)$, resulting in
\begin{align}
\Mellin_\Ssnr(\s)&
\leq \int\limits_{0}^{\infty}\underset{g(\snrmmse,\epstarget)}{\underbrace{\left((1-\epstarget)e^{
			\nd\cdot\ratecomb(\snrmmse,\nd,\epstarget)(\s-1)
		} + \epstarget\right)}}f_\Snrmmse(\snrmmse)d\snrmmse \; .
\label{eq:mellin_service_comb2}
\end{align}
The inequality is due to $\perrornorm\leq\epstarget$ as established by Lemma~\ref{lemma_csifbl}.
The optimal values of $\epstarget$ are the ones which minimize the right-hand side of (\ref{eq:mellin_service_comb2}) for some parameter $\s<1$. They can be found by minimizing the term $g(\snrmmse,\epstarget)$ over $\epstarget$ individually for each discretized value $\snrmmse$. Using observations from \cite{gursoy2013throughput}, we find:
\begin{lemma}
	\label{lemma_convex_g}
	$g(\snrmmse,\epstarget)$ is convex in $\epstarget$  for $Q(\snrmmse/\sigmac(\snrmmse))<\epstarget<1/2$ and $\s<1$. 
\end{lemma}
\begin{proof}
	See the Appendix.
\end{proof}
The convexity property establishes that the optimal $\epstarget$ for each $\snrmmse$ is unique and can be found efficiently. 
The value of $\epstarget$ then directly provides the optimal rate $\ratecomb(\snrmmse,\nd,\epstarget)$ that should be chosen by the transmitter.

Even though our numerical studies found no case where the approximation $\epstarget$ from Lemma~\ref{lemma_csifbl} was below $\perrornorm$ in (\ref{eq:yang_normalapprox}), Lemma~\ref{lemma_csifbl} relies on Assumption~\ref{assumption:sigmafbl}. 
Thus, after the optimal rate $\ratecomb(\snrmmse,\nd,\epstarget)$ has been found, the error probability $\perrornorm$ can be computed from (\ref{eq:yang_normalapprox}).

\subsection{Further Uses of the Approximation}
\label{ssec:further_uses}
Beyond our own analysis, Lemma~\ref{lemma_csifbl} could also be used for analyzing further system properties and devising other system features:
\subsubsection{Outdated CSI}
Assume that the transmitter has an outdated observation $\hat{H}_{\mathrm{old}}$ of a Rayleigh fading channel, which is related to the current value $H$ as $\hat{H}_{\mathrm{old}} = \rho H + Z$ with known $\rho$, $H\sim\mathcal{CN}(0,1)$, $Z\sim\mathcal{CN}(0,1-\rho^2)$, and $H$ and $Z$ mutually independent. Furthermore, assume that the receiver has perfect knowledge of the current channel $H$ and that is knows (e.g. from additional headers) which coding scheme was chosen by the transmitter. Then, the MMSE estimate of the channel is given as $\hat{H}_{\mathrm{old,MMSE}}=\rho\hat{H}_{\mathrm{old}}$. Replacing $\Hmmse$ in Sec.~\ref{sec:system_model} with $\hat{H}_{\mathrm{old,MMSE}}$ and $\sigmanoise^2$ with $(1-\rho^2)$ leads to the same results, i.e., the error probability $\perrornorm$ for this case can be computed through (\ref{eq:yang_normalapprox}) and approximated through Lemma~\ref{lemma_csifbl}.

\subsubsection{Power Allocation}
In certain cases, e.g., in battery-powered devices, the transmitter may need to adapt the transmit power to the channel. Power control in the finite blocklength regime was analyzed in \cite{yang2015optimum} for perfect CSI and without considering the queueing performance.
We assume now for simplicity that the channel training is always performed with the same transmit power, and that the SNR $\Snr$ and estimated SNR $\snrmmse$ are related to this training power. During the data transmission phase, the transmitted signal power is scaled by a factor $\phi>0$ and thus the SNR during the data transmission phase is changed to $\phi\Snr = \phi\snrmmse + \phi\Snrerrgauss + \phi\Snrerrdelta$. We start again by assuming, as in Assumption~\ref{assumption:sigmafbl}, that
\begin{equation}
\perrornorm \leq \Prob{\log_2\left(
	1+\phi\snrmmse+\phi\Snrerrgauss - \sigmab(\phi\snrmmse)\Ub \right) < r} \; .
\end{equation}
Following the same steps as in Sec.~\ref{ssec:finite_blocklength}, the error probability can be bounded as
\begin{align}
\perrornorm &\leq Q\left(\frac{\phi\snrmmse - \left(2^\rate-1\right)}{\sigma_{\mathrm{PA}}}\right) \; ,
\label{eq:perrorbound_pa}
\end{align}
with
$\sigma_{\mathrm{PA}}^2 = \phi^2\sigmaasq + \sigmabsq(\phi\snrmmse)$. Although there is no closed-form solution for the minimal power scaling $\phi$ such that the error probability $\perrornorm$ is below a certain target, one can quickly compute (\ref{eq:perrorbound_pa}) for different values of $\phi$ in order to determine the minimum required transmit power. 

\section{Numerical Evaluation}
\label{sec:numerics}
For numerical evaluation, Sec.~\ref{ssec:validation} addresses the accuracy of the error probability approximation from Lemma~\ref{lemma_csifbl}, especially with respect to its use in rate adaptation. 
In Sec.~\ref{ssec:validation_system_model}, we validate the accuracy of the system model itself.
In Sec.~\ref{ssec:delay_performance}, we compare the system performance without delay constraints with the performance under strict delay constraints. Finally, in Sec.~\ref{ssec:delay_arrivalrate}, we investigate the trade-off between training and data transmission time under varying delay constraints.

\subsection{Validation}
\label{ssec:validation}
\begin{figure}[t]
	\centering
	\includegraphics[width=0.98\figurewidth]{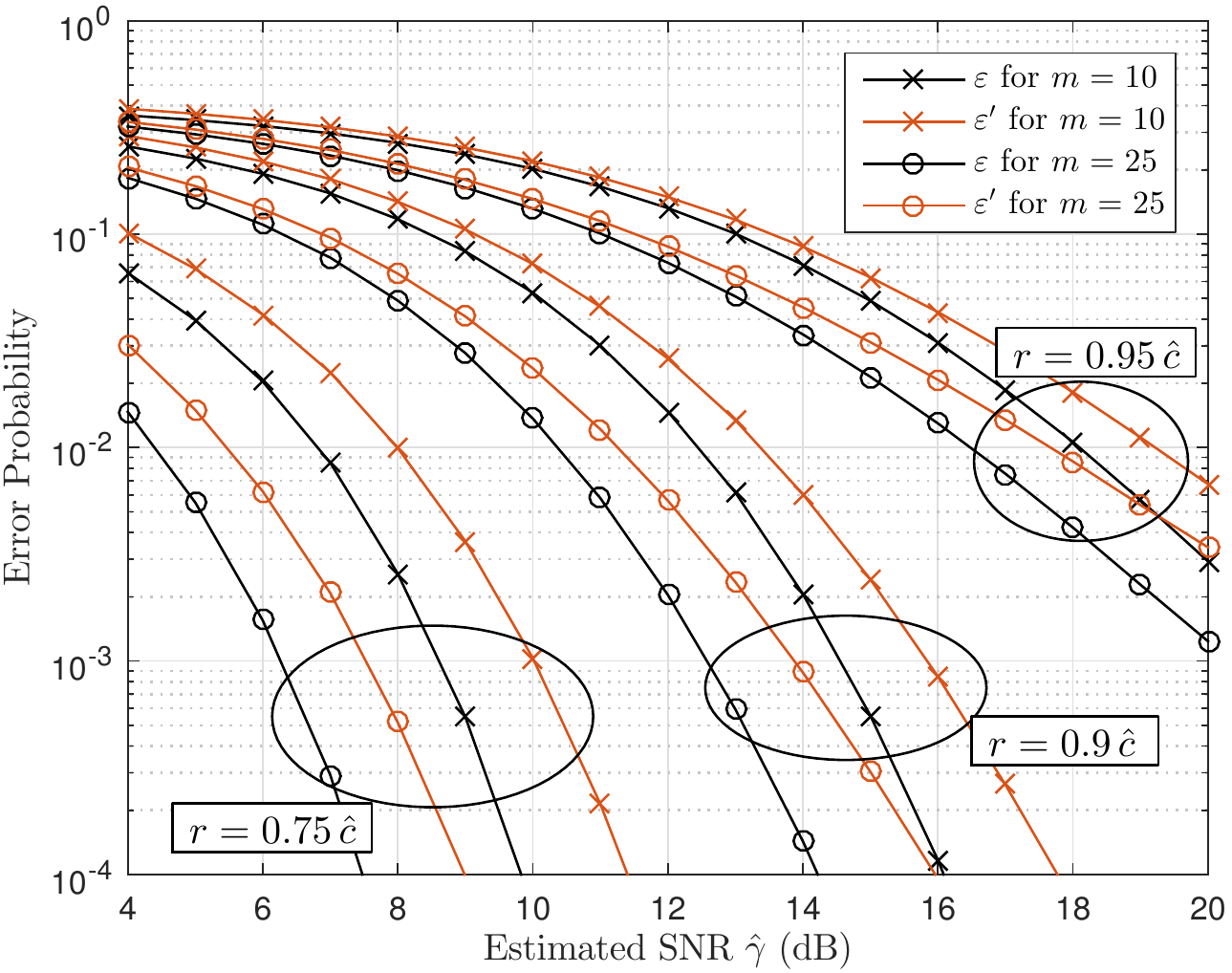}
	\caption{Error probability $\perrornorm$ and approximation/upper bound $\epstarget$ from Lemma~\ref{lemma_csifbl}  vs. estimated SNR $\snrmmse$, when the $r(\snrmmse)= \kappa \log_2(1+\snrmmse)$ with $\fraccap \in\{0.75, 0.9, 0.95\}$. Two choices of training length $\nt\in\{10,25\}$ are considered. $\nd=200$, $\snravg=15$~dB.}
	\label{fig:snrmmse_x_eps}
\end{figure}
In Fig.~\ref{fig:snrmmse_x_eps}, we compare the error probability $\perrornorm$ in (\ref{eq:yang_normalapprox}) with its upper bound $\epstarget$ from Lemma~\ref{lemma_csifbl}. The length of the training sequence is $\nt \in \{10,25\}$, the average SNR $\snravg$ is 15 dB, and the length of the data transmission phase is fixed at $\nd=200$. In this example, the rate $r(\snrmmse)$ is simply chosen as a fraction $\fraccap$ of the estimated capacity $\hat{c}=\log_2(1+\snrmmse)$, with $\fraccap\in \{0.75, 0.9, 0.95\}$.
First of all, we confirm in all cases that $\epstarget$ is indeed an upper bound on $\perrornorm$, as expected from Lemma~\ref{lemma_csifbl}. Second, even though we observe that this bound is not tight, especially when $\epstarget<10^{-2}$, it can be seen that the upper bound $\epstarget$ predicts quite accurately how much the error probability $\perrornorm$ will change when reducing the rate or when increasing the number of training symbols $\nt$. As a result, the bound/approximation $\epstarget$ may be accurate enough to decide how the rate $r(\snrmmse)$ should be adapted to the estimated SNR $\snrmmse$ and what training sequence length to choose for optimal performance.

Specifically, the optimal delay performance can be reached by iterating over the parameter $\s$ and finding the rate adaptation which minimizes $\Mellin_\Ssnr(\s)$ for each value of $\s$. 
For $\snravg=15~\mathrm{dB}$, $\nd=200$, $\nt=25$, and a specific choice\footnote{For these parameters and an arrival rate $\arrpersymb=1.4$ bits/symb., the bound $\Mfun{\s_1,w}$ for target delay $w=5$ was minimal at $\s_1\approx 0.010$, thus $\Mellin_{\Ssnr}(\s)$ must be evaluated at $\s=1-\s_1 \approx 0.99$.} of $\s=0.99$, we compare in Fig.~\ref{fig:snrmmse_x_rate} the rate adaptation based on the Corollary~\ref{corollary_csifbl_rate} as described in Sec.~\ref{ssec:optimal_ra} (red dashed curve), with a rate adaptation scheme that directly tries to minimize (\ref{eq:mellin_service_comb1}) by computing $\perrornorm$ numerically for many different values of $r$ (black solid curve, labelled as \emph{perfectRA}). We find that our proposed \emph{approxRA} scheme always selects a slightly lower rate than the \emph{perfectRA} scheme. This is due to $\epstarget$ being an upper bound to $\perrornorm$: the approximate rate adaptation scheme always overestimates the error probability and then chooses a lower rate in order to avoid too many errors. 
However, as our system model may not be accurate when $\perrornorm$ becomes extremely small (see Sec.~\ref{ssec:data_transmission}), we will from now on always restrict the \emph{approxRA} scheme to choose only values $\epstarget \ge 10^{-3}$. The resulting rates are shown in the dashed blue curve. 
We find that the difference between the \emph{perfectRA} and \emph{approxRA} schemes is very small: the value of $\Mellin_\Ssnr(\s)$ increases only slightly from $0.0291$ to $0.0294$. Furthermore, restricting \emph{approxRA} to $\epstarget \ge 10^{-3}$ is only relevant for $\snrmmse>9~\mathrm{dB}$, and has almost no effect on the value of $\Mellin_\Ssnr(\s)$, as the Mellin transform depends mostly on the behavior at low values of $\snrmmse$, where the error probabilities are much higher than $10^{-3}$ and the data rates are small.
In conclusion, even though our approximation is not tight, it can provide a nearly optimal solution to the rate adaptation problem.


Additionally, Fig.~\ref{fig:snrmmse_x_rate} shows two suboptimal rate adaptation schemes. The green dash-dotted curve shows the rate $r$ that would be chosen when the transmitter always keeps the error probability at a fixed value $\perrornorm=0.003$ for all values of $\snrmmse$. The value of $\Mellin_\Ssnr(\s)$ increases to $0.0394$ for this scheme\footnote{it is even higher for $\perrornorm\in\{0.0001,0.001, 0.002,0.01\}$ and all other values we tested}.
The second suboptimal rate adaptation scheme (violet dotted curve) is one that does not take the delay requirements into account, but optimizes the parameters to achieve the maximum expected goodput $\expgoodput$ in (\ref{eq:expected_goodput1}). This scheme favors high data rates over high reliability, and causes $\Mellin_\Ssnr(\s)$ to increase to $0.0438$. 
Due to the massive increases in $\Mellin_\Ssnr(\s)$, we suspect that the delay performance will deteriorate with both suboptimal schemes.
\begin{figure}[t]
	\centering
	\includegraphics[width=1.1\figurewidth]{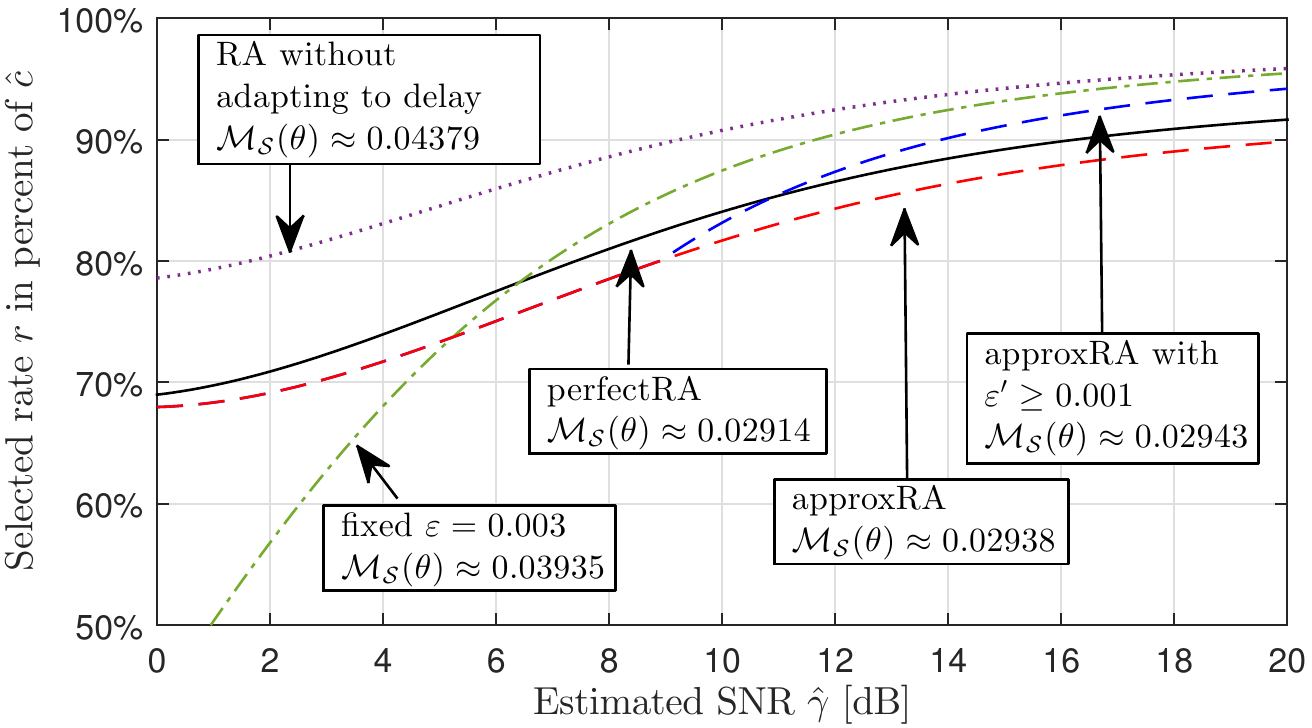}
	\caption{Choice of $r(\snrmmse)$ in percent of $\hat{c}=\log_2(1+\snrmmse)$ (with resulting values of $\Mellin_\Ssnr(\s)$) for different rate adaptation schemes. $\nt=25$, $\nd=200$, $\snravg=15$~dB, $\s=0.99$.}
	\label{fig:snrmmse_x_rate}
\end{figure}

This suspicion is confirmed by Fig.~\ref{fig:targetdelay_x_pviol_incsim}. It shows the delay violation probability $\pv(w)$, which can be obtained by simulating the queueing system with random instances of the service and arrival process, and its analytical upper bound (\ref{eq:pdelay_violation}), versus the target delay $w$ for those different rate adaptation schemes. We first note that while the upper bound (\ref{eq:pdelay_violation}) on $\pv(w)$ is not tight, which was also observed in similar works \cite{petreska2015recursive,schiessl2015delay}, the upper bound is very useful, as it not only predicts the slope of $\pv(w)$ correctly, but also predicts which parameters (here: which rate adaptation schemes) are optimal with respect to $\pv(w)$.
We observe that the delay bounds for the \emph{perfectRA} (solid black curve) and \emph{approxRA} (dashed red) are almost indistinguishable, which is in line with the results in Fig.~\ref{fig:snrmmse_x_rate}.
The difference between the two schemes in $\pv(w)$ as obtained from simulations is also not noticeable. Contrary to that, when using the suboptimal schemes, which either use fixed $\perrornorm=0.003$ or do not adapt the rate to the delay constraints, the delay violation probability $\pv(w)$ at $w=4$ degrades by nearly an order of magnitude, and this degradation is correctly predicted by the analytical bounds. 
This suggests that it is quite important to solve the rate adaptation problem optimally, taking the delay requirements into account. 

\begin{figure}[t]
	\centering
	\includegraphics[width=0.99\figurewidth]{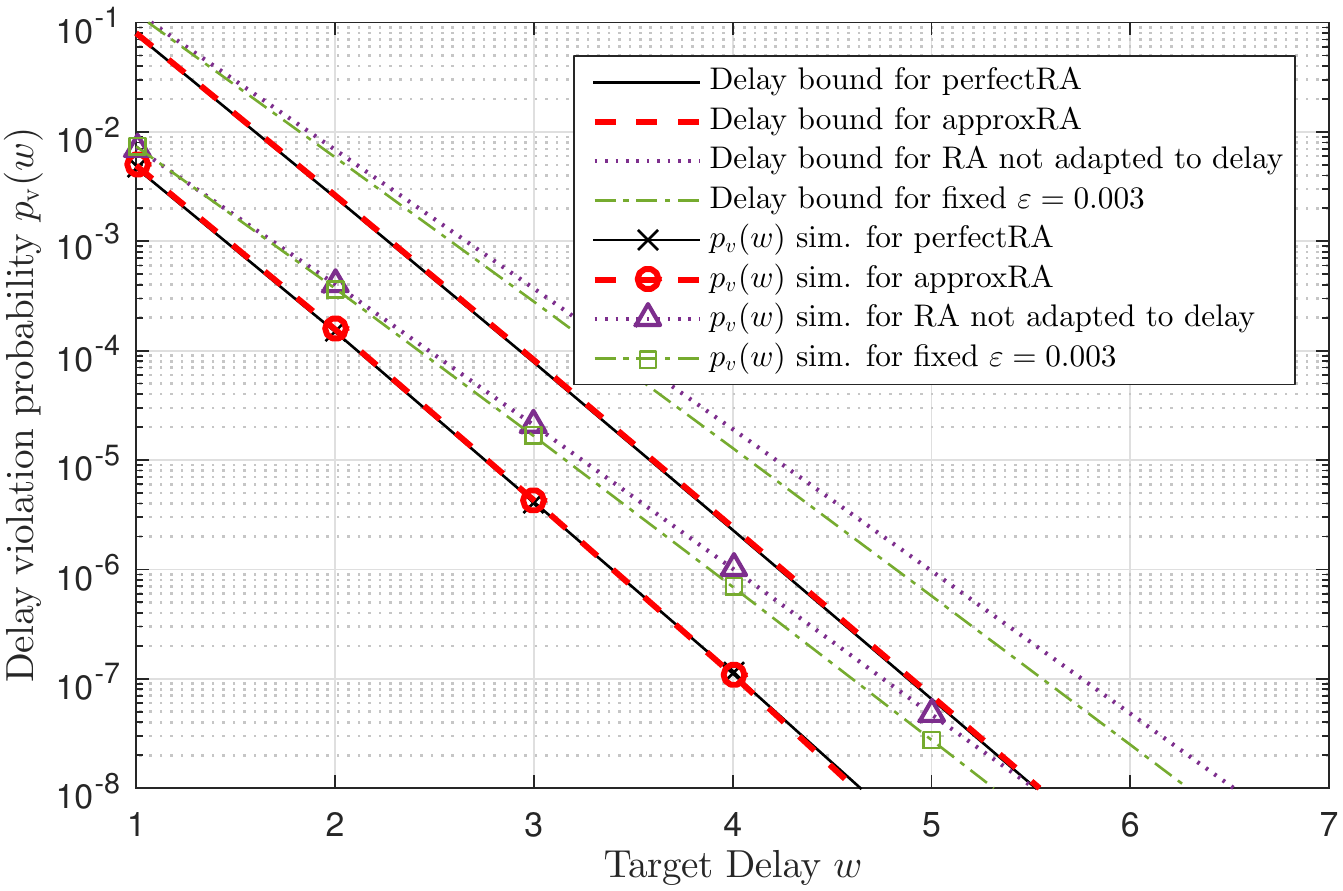}
	\caption{Target delay violation probability $\pv(w)$ (obtained from simulations over $10^{11}$ time steps) and its respective upper bound (\ref{eq:pdelay_violation}) vs. target delay $w$. $\nt=25$, $\nd=200$, $\snravg=15$~dB, arrival rate $\arrpersymb=1.4$ bits/symbol.}
	\label{fig:targetdelay_x_pviol_incsim}
\end{figure}

\subsection{Validation of the System Model}
\label{ssec:validation_system_model}
While we assumed throughout the paper that the decoding error probability is exactly equal to $\perrornorm$ in (\ref{eq:yang_normalapprox}), $\perrornorm$ is itself only an approximation, and depends on the assumption of perfect CSIR. Using \cite[Cor.~3]{yang2013quasi}, we can numerically compute a strict lower bound on the achievable rate for a given error probability. For the error probability, we set the optimal value $\epstarget$ that was found with the rate adaptation using Lemma~\ref{lemma_convex_g}. In Fig.~\ref{fig:nt_x_netcalcrate}, we compare the delay performance of a system with this achievable rate to the delay performance of the original system model. In particular, we compare the maximum possible size of arriving data packets per time slot such that the system can still guarantee different quality-of-service constraints, versus the training sequence length $\nt$. In order to guarantee a deadline of only $w=5$ slots with $\pv(w)<10^{-8}$, only 200-300 bits (depending on the training length) should arrive in each time slot. The performance (i.e., the maximum arrival size) that can be guaranteed through \cite[Cor.~3]{yang2013quasi} can be 10\% below the performance predicted from our system model. Even as the training length $\nt$ increases to $10^4$, which results in nearly perfect CSI at transmitter and receiver, this performance gap remains. However, in case the CSI becomes perfect, we know that our system model converges to earlier results by Polyanskiy et al. \cite[Thm.~54]{polyanskiy2010channel}, which were shown to be quite accurate. Therefore, while \cite[Cor.~3]{yang2013quasi} provides a strict lower bound on the performance, it is presumably not a tight bound when the CSI is nearly perfect. The following results assume again that the decoding error probability is given by (\ref{eq:yang_normalapprox}).
\begin{figure}[t]
	\centering
	\includegraphics[width=0.9\figurewidth]{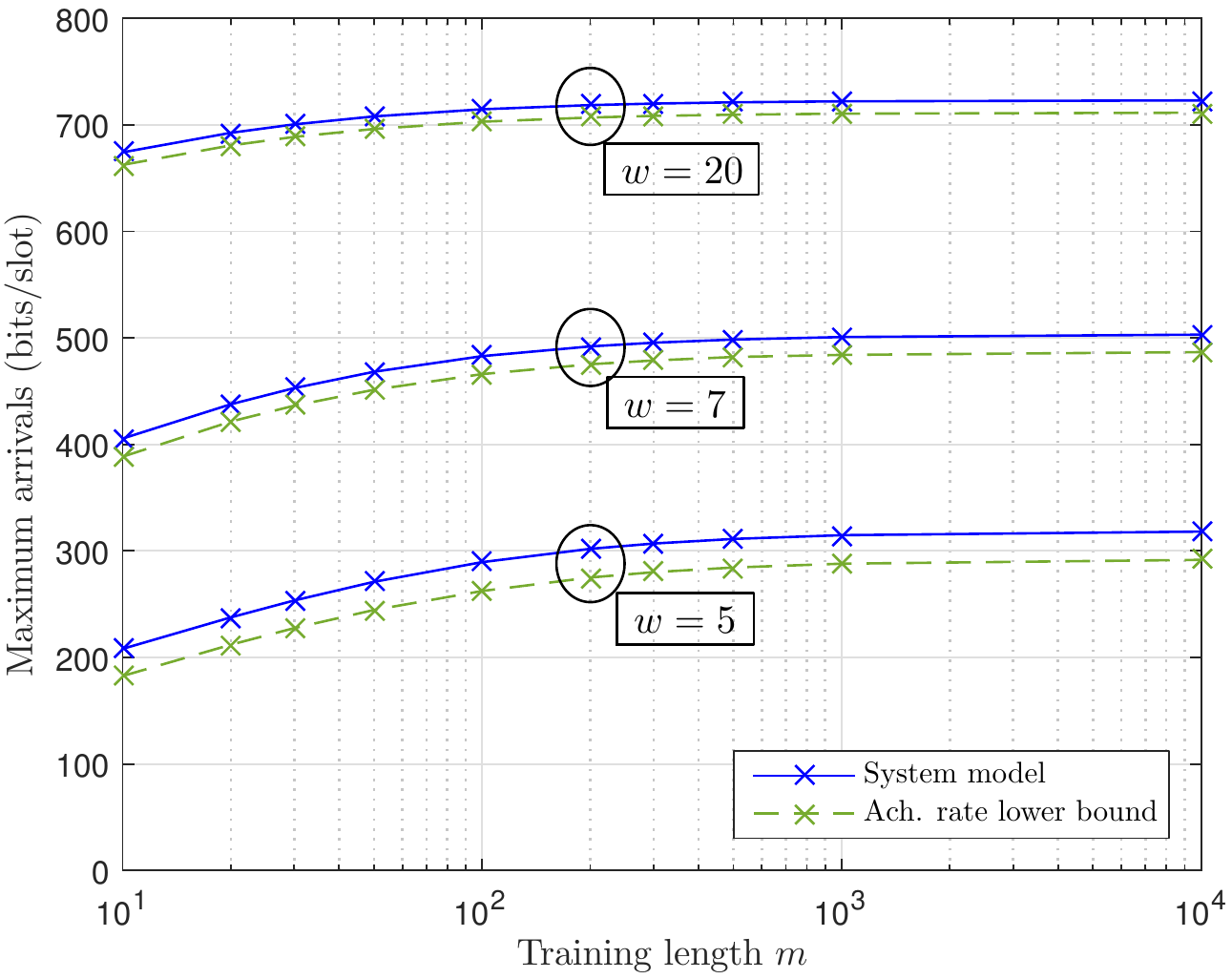}
	\caption{Maximum arrivals $\Abit$ in bits per time slot vs. training length $\nt$ for $\nd=200$, $\snravg=15~\mathrm{dB}$ such that for target delay $w\in\{5,7,20\}$ slots, the delay constraint $\pv(w)<10^{-8}$ is still satisfied.}
	\label{fig:nt_x_netcalcrate}
\end{figure}

\subsection{Performance under Delay Constraints}
\label{ssec:delay_performance}
Next, we show the impact on the performance due to imperfect CSI and finite blocklength under different delay constraints. When the system has to operate under strict delay constraints, the transmitter should generally try to achieve high reliability, which may require spending more time on channel training. In Fig.~\ref{subfig:expgoodput} we show the expected goodput $\expgoodput$, i.e. the performance when there are no delay constraints, next to Fig.~\ref{subfig:maxarrivals}, which shows the maximum supported arrival rate $\arrpersymb$ such that the system still meets strict delay constraints. In both cases, we show the performance against the average SNR $\snravg$, for different channel models and different parameters. In all cases, the total length $\ntotal=\nt+\nd$ of one time slot is 250 symbols.
First of all, Fig.~\ref{subfig:expgoodput} shows the expected goodput. The black curve (``PCSI,IBL") shows the performance for a simplified channel model where the CSI is assumed to be perfect and transmissions are error-free at rate $\rate$ equal to the capacity. The black curve marked with `+' (``PCSI,FBL") shows the performance when the CSI is still assumed to be perfect, but finite blocklength effects are taken into account. The expected goodput decreases because some transmissions fail and also because the transmitter must back off from the capacity and choose a smaller rate to get a low probability of error. When the effects of imperfect channel state information are taken into account as well (two blue curves, ``ICSI,FBL"), the performance is even lower. This is again because a backoff is required and because transmissions fail with a higher probability when the channel is unknown. Another reason is that $\nt$ symbols are used for channel estimation and thus fewer symbols are used for data transmissions, which leads to lower normalized goodput. 
In this scenario, i.e. when there are no delay constraints, we observe that the system performs better with $\nt=5$ training symbols than with $\nt=50$ over a wide range of the average SNR $\snravg$. Here, the benefits of better channel estimation at $\nt=50$ cannot compensate for the reduced length of the data transmission phase. Nevertheless, even when $\nt=50$ symbols (20\% of the time slot) are spent on training, the performance is still better than the performance of a system that does not measure the channel ($\nt=0$) and transmits at a fixed rate (red curve, ``NoCSI,FBL"). However, note that the differences are small, so this trend might change under different assumptions, for example when the feedback of CSI is not instantaneous and error-free.

\begin{figure}[t]
	\subfloat[\label{subfig:expgoodput}]{%
		\includegraphics[width=0.9\figurewidth]{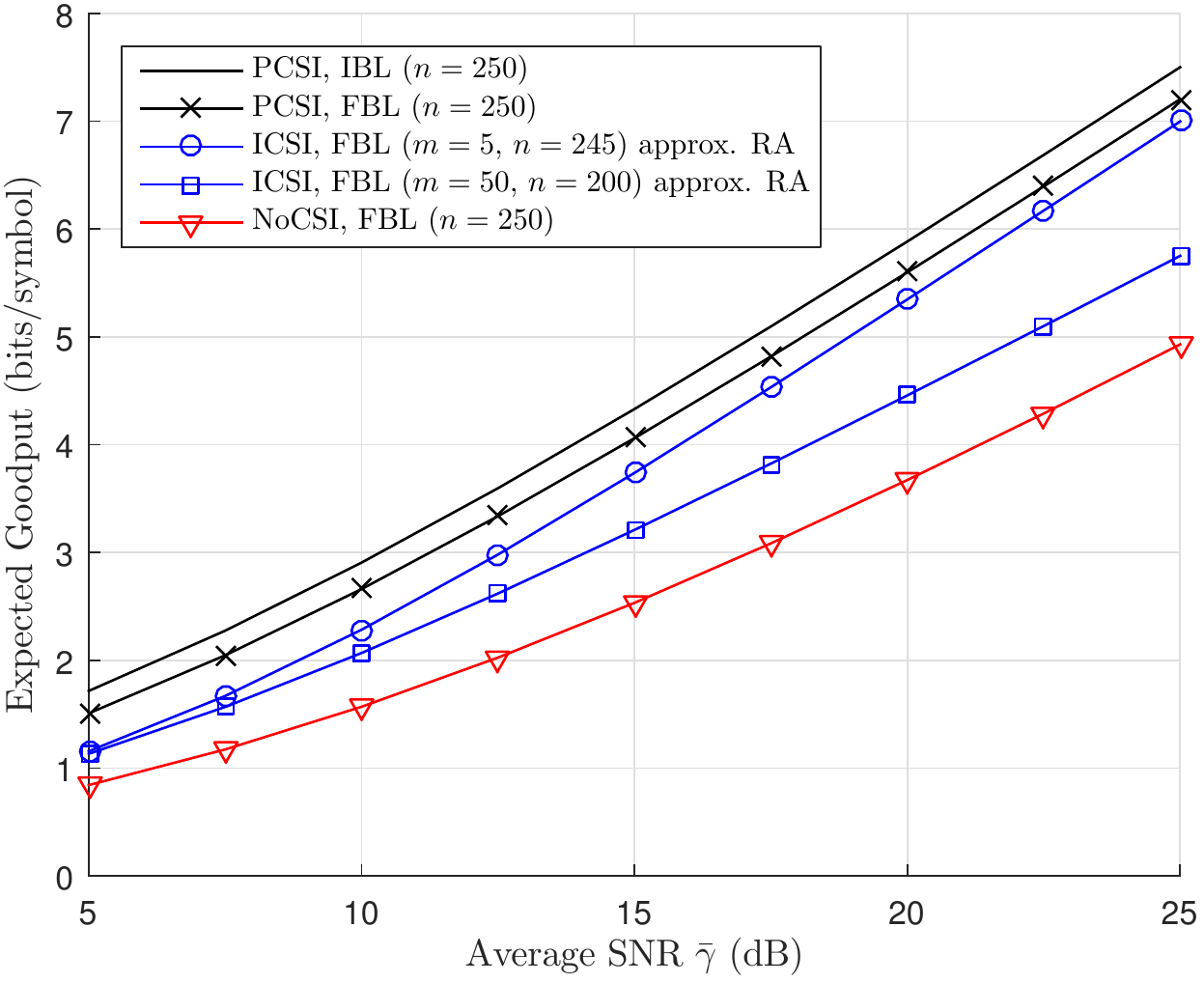}
	}
	\hfill
	\subfloat[\label{subfig:maxarrivals}]{%
		\includegraphics[width=0.9\figurewidth]{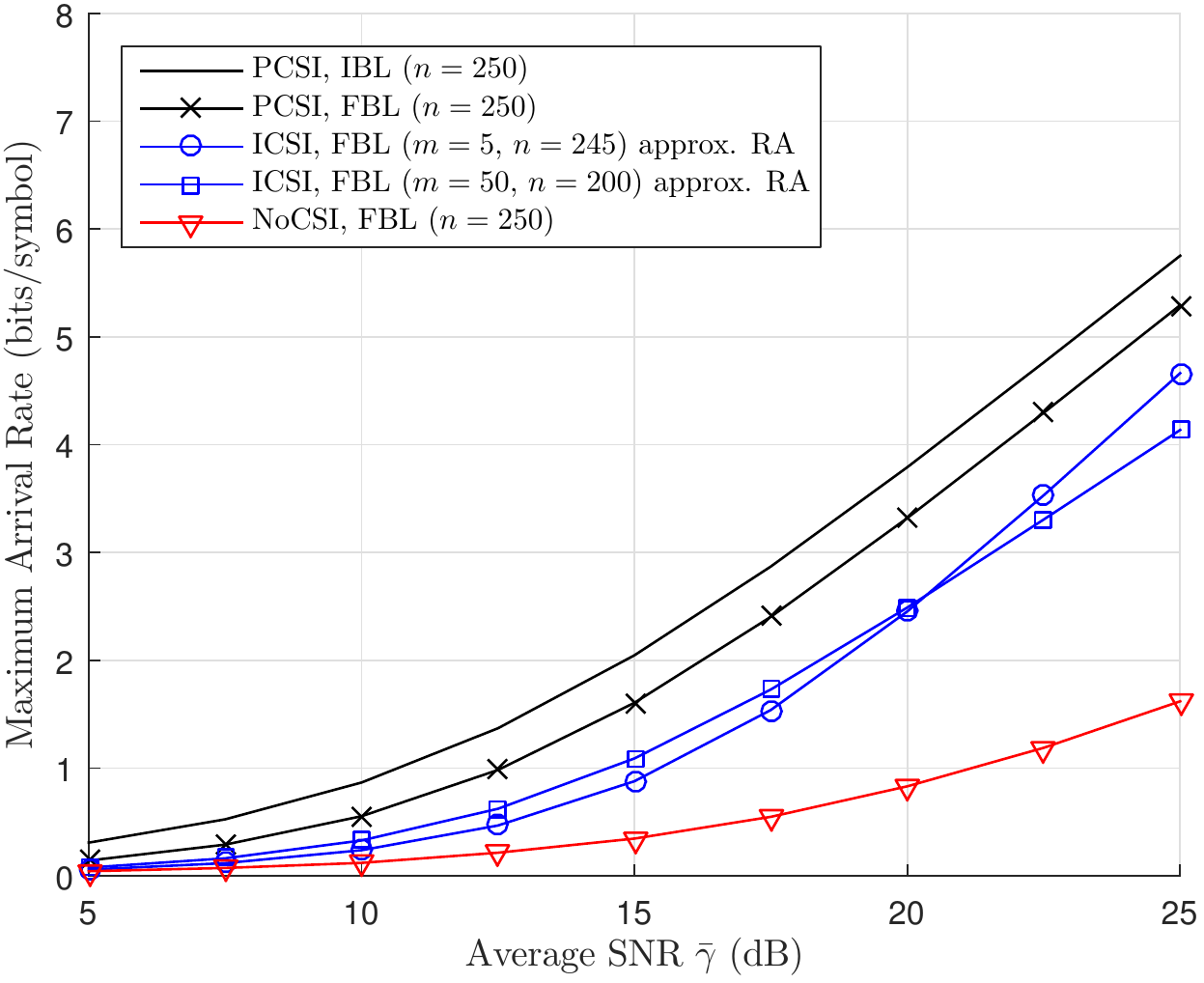}
	}
	\vspace{-2mm}
	\caption{(a) Expected goodput $\expgoodput$ vs. average SNR $\snravg$ and total slot length $\ntotal=250$. (b) Maximum arrival rate $\arrpersymb$ vs. average SNR for $\ntotal=250$ such that for target delay $w=5$ slots, $\pv(w)<10^{-8}$.}
	\label{fig:dummy}
\end{figure}

Fig.~\ref{subfig:maxarrivals} shows the maximum arrival rate $\arrpersymb$ per symbol if the upper bound on the delay violation probability $\pv(w)$ for a target delay of $w=5$ slots should not be higher than $10^{-8}$, for the same system parameters as in Fig.~\ref{subfig:expgoodput}. It can be seen that it is very difficult for the system to meet these target requirements with imperfect CSI when the average SNR is low. The requirements can only be satisfied by choosing an arrival rate $\arrpersymb$ that is significantly lower than the expected goodput. 
At low SNR, we also see that using $\nt=50$ training symbols now leads to better performance than $\nt=5$. This means that under strict delay requirements, the higher reliability gained through better channel estimation is more beneficial than additional data transmission time. The trade-off between $\nt$ and $\nd$ depends on the delay constraints.
Furthermore, we observe that the fixed rate transmission scheme performs significantly worse than the schemes adapting the rate to the imperfect measurement. Thus, rate adaptation seems to be beneficial especially under tight delay constraints.

\subsection{How much time to spend on training?}
\label{ssec:delay_arrivalrate}
\begin{figure}[t]
	\centering
	\includegraphics[width=0.9\figurewidth]{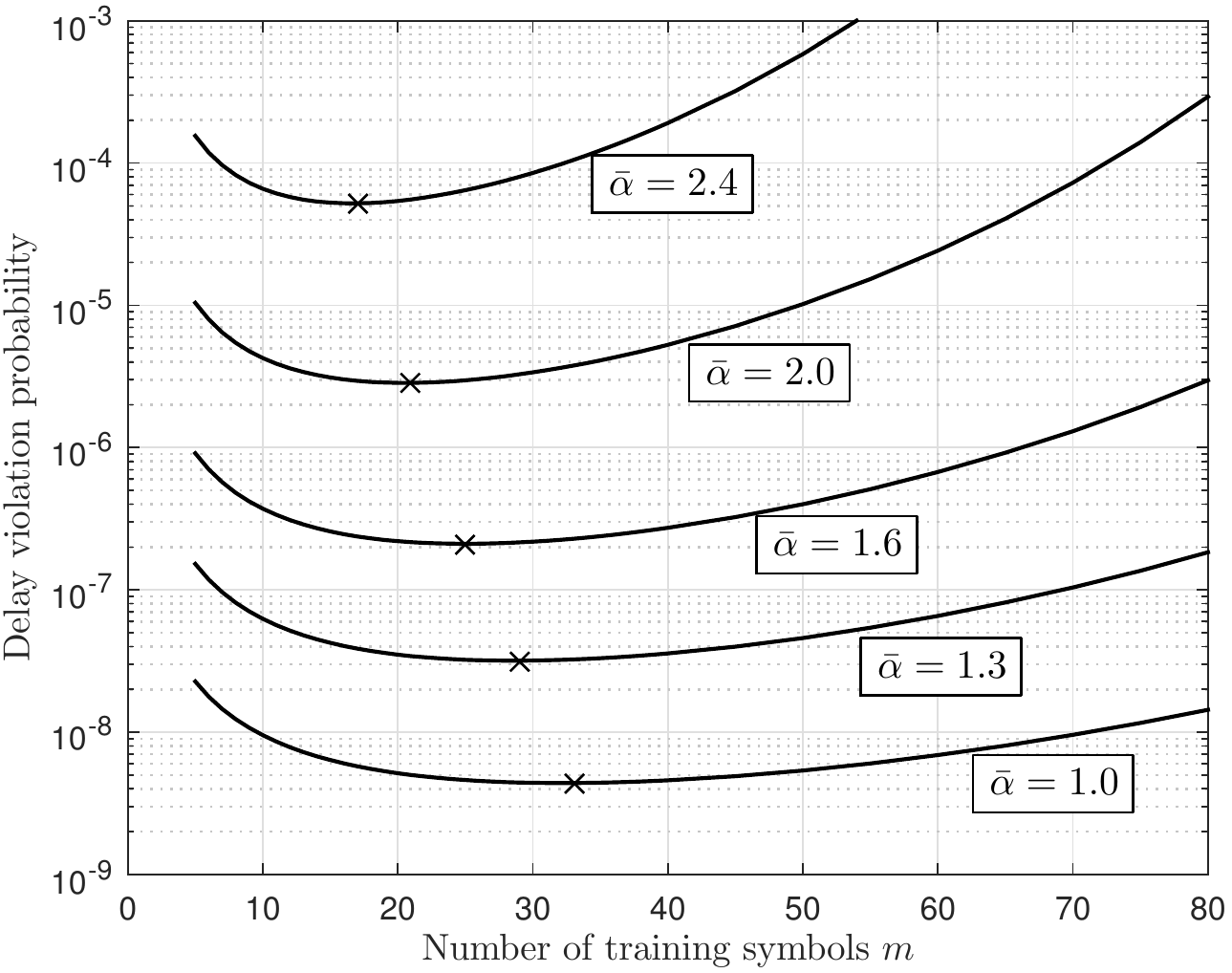}
	\caption{Bounds on the delay violation probability $\pv(w)$ vs. number of training symbols $\nt$ for target delay $w=5$, average SNR $\snravg=15$, $\ntotal=250$ and different arrival rates (in bits per symbol). The minimum point (optimal $\nt$) is marked with `x'.}
	\label{fig:ntrain_x_pviol}
\end{figure}
The previous Fig.~\ref{subfig:maxarrivals} implies that with $\ntotal=250$ and $\snravg=15$~dB, it is possible to meet the QoS target $\pv(w=5)<10^{-8}$ when $\arrpersymb\approx 1.0$ bits per symbol and $\nt=50$. For $\nt=5$, the performance is worse, but what is the optimal value of $\nt$? Fig.~\ref{fig:ntrain_x_pviol} shows the delay bound for $\ntotal=250$, $\snravg=15$~dB and different values of $\arrpersymb$ versus the training length $\nt$. For $\arrpersymb=1.0$, the smallest delay violation probability is obtained at $\nt=33$ training symbols, but the performance remains similar for $\nt$ between 20 and 50. On the other hand, when the arrival rate is increased, fewer training symbols should be used; the delay violation probability easily increases by an order of magnitude when the transmitter chooses too many (e.g. $\nt=50$) training symbols.

Finally, Fig.~\ref{fig:targetdelay_x_ntopt} shows the optimal number of training symbols $\nt$ for different delay requirements and different SNR levels. Here, we require that for all values of the target delay $w$, the bound on the delay violation probability is $\pv(w)<10^{-8}$. The optimal $\nt$ is then defined as the value of $\nt$ for which the system can support the highest arrival rate $\arrpersymb$ while still satisfying the delay constraints. We observe that when the delay requirements become very strict, a large fraction of the available resources must be spent on training. On the other hand, when the delay requirements become more relaxed, the optimal value of $\nt$ is much smaller. Lastly, when the SNR decreases, channel estimation becomes less reliable, and more training symbols are required.
\begin{figure}[t]
	\centering
	\includegraphics[width=0.9\figurewidth]{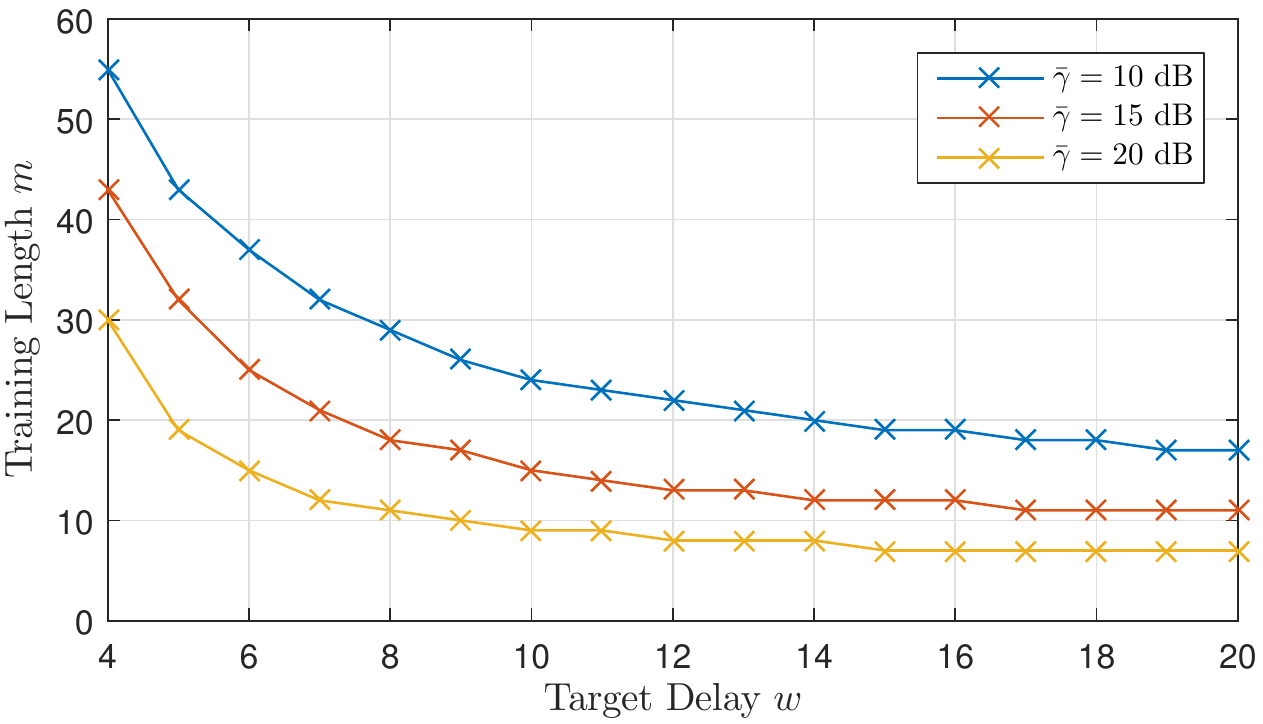}
	\caption{Optimal value of training length $\nt$ when $\ntotal=250$ against the target delay $w$, for different average SNR $\snravg \in\{10, 15, 20\}~\mathrm{dB}$.}
	\label{fig:targetdelay_x_ntopt}
\end{figure}


\section{Conclusions and Future Work}
\label{sec:conclusions}
In this work, we studied the joint impact of rate adaptation with imperfect CSI at the transmitter and finite blocklength channel coding on the delay performance of a wireless communication system. 
Based on stochastic network calculus, we found that the transmitter must adapt the rate not only to the channel estimate, but must also consider the delay requirements. 
In order to find the optimal rate adaptation, we developed a closed-form approximation for the error probability due to imperfect CSI and finite blocklength, which is invertible and can be used beyond this specific analysis.
Using this approximation, the optimal rate selection becomes a convex problem.
After validating various aspects of our work, we showed numerically that rate adaptation typically outperforms approaches that are channel-agnostic (i.e. do not rely on CSI at the transmitter). 
Furthermore, it could be shown that the optimal choice of training length depends strongly on the average SNR as well as on the delay and reliability requirements of the application. Making an optimal choice may reduce the delay violation probability of the system by an order of magnitude.
A possible extension of this work relates to multi-antenna systems, which can offer higher reliability at the cost of even more channel estimation, while CSI at the transmitter could also be used for beamforming.




\appendix

\section{Proof of Lemma~\ref{lemma_convex_g}}
\label{appendix:lemma_convex_g}
To show convexity, we show that for fixed $\snrmmse$, the second derivative of $g_1(\epsappendix)=g(\snrmmse,\epsappendix)$ is strictly positive for $\s<1$:
\begin{align}
g_1(\epsappendix) &= (1-\epsappendix)
e^{\nd\cdot\ratecomb(\snrmmse,\nd,\epsappendix)(\s-1)}
+ \epsappendix\\
&=(1-\epsappendix)
(1+\snrmmse - \sigmac(\snrmmse) Q^{-1}(\epsappendix))^{\frac{\nd}{\ln 2}(\s-1)}
+ \epsappendix\\
&=(1-\epsappendix)
(a - b Q^{-1}(\epsappendix))^{c}
+ \epsappendix
\end{align}
with constants $a,b>0$ and $c<0$. Due to $\epsappendix>Q(\snrmmse/\sigmac(\snrmmse))$, we have $\ratecomb(\snrmmse,\nd,\epsappendix)>0$ and $(a - b Q^{-1}(\epsappendix))>1$.
The first derivative is given by 
\begin{align}
\dot{g}_1(\epsappendix) =& (1-\epsappendix)
c(a - b Q^{-1}(\epsappendix))^{c-1}(- b \dot{Q}^{-1}(\epsappendix))
-(a - b Q^{-1}(\epsappendix))^{c} + 1.
\end{align}
The second derivative is given by
\begin{align}
\ddot{g}_1(\epsappendix) =& (1-\epsappendix)
c(a - b Q^{-1}(\epsappendix))^{c-1}(- b \ddot{Q}^{-1}(\epsappendix))\nonumber\\
& + (1-\epsappendix)
c(c-1)(a - b Q^{-1}(\epsappendix))^{c-2}(- b \dot{Q}^{-1}(\epsappendix))^2\nonumber\\
&-
2c(a - b Q^{-1}(\epsappendix))^{c-1}(- b \dot{Q}^{-1}(\epsappendix)).
\end{align}
From \cite{gursoy2013throughput}, the derivatives of the inverse Q-function are:
\begin{align}
\dot{Q}^{-1}(\epsappendix) &= -\sqrt{2\pi}e^{\frac{Q^{-1}(\epsappendix)^2}{2}}\\
\ddot{Q}^{-1}(\epsappendix) &= 2\pi Q^{-1}(\epsappendix) e^{Q^{-1}(\epsappendix)^2}
\end{align}
Thus, for $\epsappendix<\nicefrac{1}{2}$, $\dot{Q}^{-1}(\epsappendix) < 0$ and $\ddot{Q}^{-1}(\epsappendix) > 0$, and therefore $\ddot{g}_1(\epsappendix) > 0$.

\bibliographystyle{IEEEtran}
\bibliography{fbl_impcsi_main}


\end{document}